\documentclass[sigconf,10pt,balance=false]{acmart}

\copyrightyear{2026}
\acmYear{2026}
\setcopyright{cc}
\setcctype{by}
\acmConference[MobiCom '26]{The 32nd Annual International Conference on Mobile Computing and Networking}{October 26--30, 2026}{Austin, TX, USA}
\acmBooktitle{The 32nd Annual International Conference on Mobile Computing and Networking (MobiCom '26), October 26--30, 2026, Austin, TX, USA}
\acmDOI{10.1145/3795866.3844487}
\acmISBN{979-8-4007-2505-0/2026/10}


\author{Farzad Mehri}
\affiliation{%
  \institution{NC State University}
  \city{Raleigh}
  \country{USA}}
\email{fmehri@ncsu.edu}

\author{Dara Ron}
\affiliation{%
  \institution{NC State University}
  \city{Raleigh}
  \country{USA}}
\email{dron@ncsu.edu}

\author{Nishanth Sastry}
\affiliation{%
  \institution{University of Surrey}
  \city{Guildford}
  \country{UK}}
\email{n.sastry@surrey.ac.uk}

\author{Satyaki Roy}
\affiliation{%
  \institution{University of Alabama in Huntsville}
  \city{Huntsville}
  \country{USA}}
\email{sr0215@uah.edu}

\author{Vijay K. Shah}
\affiliation{%
  \institution{NC State University}
  \city{Raleigh}
  \country{USA}}
\email{vijay.shah@ncsu.edu}

\renewcommand{\shortauthors}{Mehri, Ron, Sastry, Roy, Shah.}

\usepackage[english]{babel}
\usepackage{blindtext}

\usepackage{amsthm}
\usepackage{dblfloatfix}
\usepackage{algorithm}
\usepackage{algorithmic}

\usepackage{amsmath}
\usepackage{wrapfig}
\usepackage{graphicx}
\usepackage{textcomp}
\usepackage{xcolor}
\usepackage{multicol}
\usepackage{enumitem}
\usepackage{soul}
\usepackage{subcaption}
\usepackage{url}

\usepackage{tikz}
\usepackage{circledtext}
\usepackage{pifont}
\usepackage{svg}

\newcommand{\T}{\mathcal{T}}

\newcommand{\C}{C}

\newcommand{\NSites}{$25$\xspace}
\newcommand{\PImprovement}{$90.68$\xspace}

\newcommand{\sysabb}{\textsc{SkyShare}\xspace}
\newcommand{\algname}{\textsc{SkySched}\xspace}

\newcommand{\bstep}[1]{
  \tikz[baseline=(char.base)]{
    \node[shape=circle, fill=black, inner sep=1pt] (char)
    {\color{white}\scriptsize\bfseries #1};
  }%
}

\begin{document}

\title[\sysabb]{\sysabb: Constellation-wide Sky Sharing for LEO-Radio Astronomy Coexistence}

\begin{abstract}

Rapidly growing low-Earth-orbit (LEO) constellations increasingly operate in the spectrum shared with radio astronomy services (RAS), creating escalating interference risks for sensitive scientific observations. Existing mitigation mechanisms rely on reactive beam steering, or avoidance near observatories but fail to account for aggregate sidelobe emissions—leading to residual interference and substantial, unnecessary capacity loss. We present \sysabb, a constellation-wide sky-sharing system that enables predictive, interference-aware spot beam scheduling to protect radio astronomy while preserving network coverage. \sysabb integrates high-fidelity orbital prediction with International Telecommunication Union (ITU)-compliant Equivalent Power Flux Density (EPFD) modeling, and real-time observatory data via Operational Data Sharing (ODS) to jointly optimize beam--cell assignments over observation windows. To make constellation-scale coordination tractable, we introduce a concept of EPFD-budgeted Region-of-Interest (RoI) that bounds residual sidelobe interference while confining optimization to a minimal, provably sufficient set of cells. Building on RoI, we formulate LEO--RAS coexistence as a scalable scheduling problem and design \algname, a flow-based algorithm that is optimal in special cases and yields scalable near-optimal solutions in the general NP-hard setting. \sysabb operates entirely in the control plane and requires no satellite hardware changes. Using real Starlink constellation geometries, we evaluate \sysabb across \NSites{} single-dish, Ku-band RAS sites worldwide. Compared to Starlink boresight avoidance, \sysabb reduces unserved cells by up to \PImprovement{}\% while remaining within EPFD limits.
\end{abstract}


\maketitle

\vspace{-0.1in}
\section{Introduction}\label{sec:intro}

Large-scale deployment of LEO communication satellites has become a critical component of global internet connectivity \cite{zhang2022leo}, with commercial systems such as Starlink, Kuiper, and OneWeb now operating thousands of satellites and providing broadband access to rural, remote, and underserved regions~\cite{leoamazonleo_updates_2026,leoey_leo_2025,leooneweb_constellation_2026,leostarlink2025progress}. However, the ubiquitous availability of LEO satellites poses a threat to Radio Astronomy Services (RAS), i.e., while traditionally, the geographical distance of RAS sites from urban areas protected against terrestrial radio frequency interference (RFI), now every location on the earth is susceptible to RFI from LEO constellations \cite{nhan2025ods}.

RAS (or Radio Telescopes) are designed to detect extremely weak cosmic signals and therefore require stringent protection from interference sources \cite{ITU-RA769-2}.  As LEO constellations continue to scale, ensuring coexistence with radio astronomy without compromising global connectivity (or scientific discovery) has become an urgent operational challenge.

\vspace{-0.1in}
\subsection{Current Practice and Its Limitations}\label{sec:current_practice_limitations}

To protect RAS, satellite operators, including SpaceX \cite{Starlink}, have begun collaborating with radio astronomy organizations through the Operational Data Sharing (ODS) framework \cite{NRAOODS2025}, which provides real-time information about RAS locations, observation schedules, pointing directions, and frequency bands. Leveraging this information, existing coexistence mechanisms primarily rely on \emph{reactive, threshold-based beam avoidance} in conjunction with exclusion zones. A representative example is Starlink Telescope Boresight Avoidance (TBA) approach~\cite{Starlink}, in which satellites disable or steer spot beams away when their angular separation from a RAS boresight falls below respective predefined thresholds \cite{Starlink,10539143}.
While these methods are effective at preventing the most severe interference events, they suffer from two limitations: 

First, they fail to account for \emph{aggregate sidelobe interference} (See \S\ref{sec:interference}). Sidelobe emissions (Figure~\ref{fig:interference-level}) from many satellites can still accumulate and increase the RFI level, especially for smaller RAS antennas (or at lower frequencies). As a result, existing reactive, threshold-based beam avoidance mechanisms fundamentally limit the achievable RAS protection level as the scale of users and satellites increases, compromising the quality of RAS observations.

Second, they can create \emph{persistent coverage gaps} due to reactively disabling beams or redirecting them. 
To maintain the required protection level in worst-case scenarios, satellite operators have to disable service to certain regions near RAS sites, which limits service over large geographic areas, disproportionately affecting rural, remote users who rely on satellite Internet access. With over a hundred active radio astronomy sites worldwide~\cite{goastronomyRadioTelescopes}, exclusion-zone-like protection affects a growing number of users as constellations and usage grow. Figure~\ref{fig:tba_char} shows this \textit{limitation}:
With RFI quantified by Equivalent Power Flux Density (EPFD) metric (see \S\ref{sub:interference-model}), configuring the parameters of the TBA 
(detailed in \S\ref{sub:tba})
cannot arbitrarily reduce EPFD without sacrificing coverage. As a result, protecting observations often requires disabling a significant set of surrounding cells. And even then, the remaining EPFD is dominated by the aggregate of many satellites’ sidelobes (see \S\ref{sub:sidelobes}), so further exclusion of cells from receiving service only results in diminishing EPFD reduction (Figure~\ref{fig:tba_char}), with the residual interference limiting achievable RAS sensitivity (see \S\ref{subsec:impact_ras})

\begin{figure}[t]
    \centering
    \begin{subfigure}[t]{0.49\linewidth}
        \includegraphics[width=\linewidth]{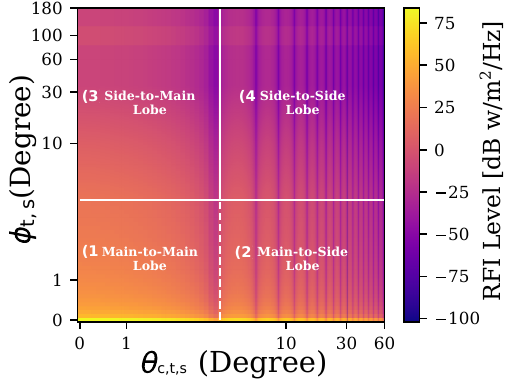}
        \caption{RFI levels under varying angle between satellite $s$, RAS ($\phi_{t,s}$), and ground cell $c$, s, RAS ($\theta_{c,t,s}$)  at a time $t$.}
        \label{fig:interference-level}
        \vspace{-0.1in}
    \end{subfigure}
    \hfill
    \begin{subfigure}[t]{0.47\linewidth}
        \includegraphics[width=\linewidth]{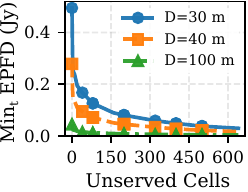}
        \caption{Spectral EPFD vs. unserved cell under TBA sweeps for three RAS diameters \textbf{$(1 \;\mathrm{Jy}: 10^{-26}\mathrm{W\,m^{-2}\,Hz^{-1}}).$}}
        \vspace{-0.1in}
        \label{fig:tba_char}
    \end{subfigure}
    \caption{RAS interference dynamics and TBA characterization.}
    \label{fig:tba_combined}
    \vspace{-0.2in}
\end{figure}

\textbf{Key Insight.} These limitations stem from a common root cause: radio astronomy protection is treated as a \emph{local, reactive, rule-based} problem rather than a \emph{constellation-wide} one. Existing approaches \cite{Starlink,nhan2025ods, nhan2024spectrum} reason about satellites independently and act only after a satellite enters a restricted region, ignoring aggregate interference and failing to exploit spatial and temporal diversity.
As shown in \S\ref{sec:interference}, multiple satellites can typically serve the same ground cell at any moment, yet their interference at a RAS varies widely due to geometry, antenna patterns, sidelobes, and relative motion. This diversity allows reallocating service across satellites to preserve coverage while reducing aggregate RFI. Our key insight is that effective LEO–RAS coexistence requires \emph{predictive, constellation-wide coordination}. Because satellite orbits are highly predictable and radio astronomy observations are scheduled in advance, satellite–cell assignments can be planned proactively, before harmful interference arises. As illustrated in Figure~\ref{fig:intro}, a coordinated scheduler maintains connectivity by reassigning spot beams at time $t$. Jointly accounting for future satellite positions, observation schedules, and aggregate RFI transforms RAS protection from reactive beam shutdowns into a principled scheduling problem.

\begin{figure}
    \centering
    \includegraphics[width=0.9\linewidth]{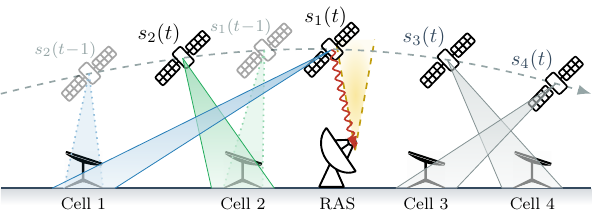}
    \vspace{-0.1in}
    \caption{
    Satellites
    adjust their associations with ground cells over time. At time~$t$, $s_1$ lies within the RAS receiving beam, and triggers handover for all shown cells to protect RAS.}
    \label{fig:intro}
    \vspace{-0.1in}
\end{figure}

\subsection{Design and Contributions}
In this paper, we present \sysabb, a predictive control-plane system for constellation-wide LEO--RAS spectrum sharing. \sysabb combines satellite ephemerides, real-time RAS metadata from ODS, and cell demand to proactively schedule satellite spot beams under time-averaged EPFD constraints. It jointly assigns satellite--cell associations to minimize coverage loss while satisfying interference limits. We formulate this problem as a mixed-integer optimization and show it is NP-hard. To enable scalable operation, we design \algname, a flow-based scheduler that yields optimal solutions in special cases and efficient near-optimal solutions in general settings.

In summary, this paper makes the following contributions:

\noindent $\bullet$ We present \sysabb, the first control-plane system for predictive, constellation-wide spot-beam scheduling that protects RAS while preserving LEO network coverage.

\noindent $\bullet$ We develop an ITU-compliant EPFD interference model capturing main-lobe and aggregate sidelobe effects, and introduce an EPFD-budgeted Region-of-Interest (RoI) abstraction that bounds sidelobe aggregation while restricting optimization to a minimal, provably sufficient set of cells.

\noindent $\bullet$ Using EPFD-budgeted RoI, we formulate constellation-wide LEO--RAS scheduling under spot-beam limits, heterogeneous cell demand, and EPFD constraints, show the problem is NP-hard, and design \algname, a scalable scheduler with min-cost-flow optimality in special cases, and near-optimal solutions (via feasibility repair) in the general setting.

\noindent $\bullet$  Using real Starlink Gen2-mini constellation  (obtained from CelesTrak)
and $25$ single-dish, Ku-band RAS sites worldwide, we show that \sysabb reduces unserved cells by up to \PImprovement{}\% compared to state-of-the-art baselines while maintaining EPFD compliance.

\section{Background and Preliminaries}
\label{sec:background}
This section provides background on LEO satellite spot beam operation, RAS observations, protection criteria, and the state-of-the-art TBA approach employed by SpaceX Starlink.

\vspace{-0.1in}
\subsection{LEO Satellites and Spot Beam Operation}

Modern LEO satellite constellations, such as Starlink and Kuiper~\cite{kim2026satellite, FCC2148, Kuiper2019}, employ electronically steerable phased-array antennas capable of forming and steering multiple narrow spot beams simultaneously. These antennas enable satellites to dynamically allocate downlink capacity across many ground locations without mechanical motion. We model Starlink antenna as a 2D rectangular phased array (\textbf{Appendix~\ref{app:star-antenna}}).

Let $S=\{s_i\}_{i=0}^{|S|-1}$ denote the set of LEO satellites in the constellation. Each satellite can form up to $N^{sb}$ independent spot beams at any time, limiting the number of ground locations it can serve simultaneously.

We adopt a \emph{quasi-Earth-fixed} cell model \cite{wang2024systematic, fu2023satellite, lin2022path}, in which each spot beam illuminates a fixed ground footprint, referred to as a \emph{cell}, for a short duration when the satellite passes overhead by continuously steering spot beams, as adapted by Starlink \cite{wang2024systematic}. We denote the set of all cells worldwide as $C^{all}=\{c_i\}_{i=0}^{|C^{all}|-1}$. We assume each cell is served by at most one spot beam at a time; multi-beam service is discussed in \S\ref{sec:discussion}.

LEO satellites adjust transmit power dynamically to satisfy regulatory Power Flux Density (PFD) limits on the ground, which are set to protect terrestrial services\footnote{ITU limits the maximum power flux density (PFD) on the ground for LEO ku-band Satellite operation \cite{ITU-R_SF.1482,ITUMaxPower}.}. Consequently, when a satellite is visible to a cell, the achievable downlink capacity is largely determined by allocated bandwidth and receiver characteristics once minimum visibility constraints are satisfied. In this work, we therefore model spot beam capacity as fixed whenever a satellite is visible to a cell (\textbf{more details in Appendix \ref{app:spot beam-cap}}).

In practice, LEO satellites periodically reconfigure spot beams over a scheduling window $w$ (e.g., Starlink uses 15\,s intervals~\cite{tanveer2023making}). We denote the \textbf{window size} by $N_w$, during which satellite--cell assignments remain fixed.

\begin{figure}[t]
    \centering
    \begin{subfigure}[t]{0.49\linewidth}
    \includegraphics[width=\linewidth]{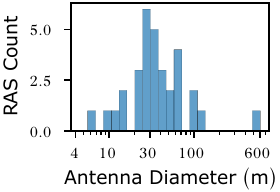}
        \caption{Histogram of RAS antenna diameter.}
        \label{fig:antenna_hist}
    \end{subfigure}
    \hfill
    \begin{subfigure}[t]{0.49\linewidth}    \includegraphics[width=\linewidth]{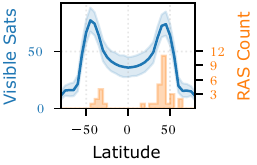}
        \caption{Histogram of RAS count, visible satellites vs latitudes.}
        \label{fig:ras_latmpact}
    \vspace{-0.25in}
    \end{subfigure}
    \caption{Characterization of RAS Sites Worldwide}
    \vspace{-0.2in}
\end{figure}

\subsection{RAS Telescope Operation} \label{sec:ras_telescope_operation}

More than one hundred RAS telescopes are deployed worldwide~\cite{goastronomyRadioTelescopes,wiki_radio_telescopes}. Figure~\ref{fig:antenna_hist} shows that these sites exhibit diversity in antenna size, ranging from a few meters to over $100$ meters. Therefore, RAS receiver characteristics vary widely~\cite{sokolowski2022ska}. Figure~\ref{fig:ras_latmpact} illustrates the geographic distribution of RAS sites and the corresponding variation in satellite visibility across latitudes. 
RAS observations are typically scheduled in advance, and their operational metadata can be shared with LEO satellite operators through Operational Data Sharing (ODS) as discussed below.

\vspace{-0.1in}
\paragraph{\textbf{2.2.1. Operational Data Sharing (ODS) Framework}}
\label{sub:ods-and-timing}
ODS is an automated, self-reporting framework developed by National Radio Astronomy Observatory (NRAO) in which a radio observatory publishes near real-time telescope operational metadata
including telescope location, observation direction, schedules, and operation frequency bands 
into a secure database that participating satellite operators can access via a REST API, enabling near real-time adaptation of satellite downlink for spectrum coexistence \cite{nhan2025ods, NRAOODS2025}.

A radio observation (RO) can be represented as $RO = \{L_R, E_R, \Omega_O, D_O, F_R, \Delta F, T_{Start}, T_{End}\}$ where $L_R$ and $E_R$ denote RAS telescope location and elevation, $\Omega_O$ and $D_O$ are right ascension and declination, $F_R$ and $\Delta F$ denote operating frequency and bandwidth, and $T_{Start}$ and $T_{End}$ define the scheduled observation duration.

The \textbf{observation duration} is $\tau = T_{End}-T_{Start}$. For example, Green Bank Telescope spectral observations commonly use integration times of 1\,s or longer~\cite{nhan2024spectrum}, while ITU-R recommendations specify protection criteria over intervals extending to thousands of seconds~\cite{ITU-RA769-2}.

Given the fast-moving LEO satellite network, let us discretize $\tau$ into multiple time slots, with the total number of slots given by $|\mathcal{T}| = \tau/\Delta t$, where $\Delta t$ denotes the slot length. A small slot length $\Delta t$ (typically 1 s, as in \cite{di2023unintended}) is used to capture the dynamics of LEO satellite–RAS interference.

\vspace{-0.1in}
\paragraph{\textbf{2.2.2. Time-Averaged EPFD}} Each time slot $\Delta t$ captures instantaneous interference from moving satellites. Aggregate interference over the full observation duration $\mathcal{T}$ is evaluated using the time-averaged equivalent power flux density (EPFD), following ITU-R S.1586-1~\cite{ITUR_S.1586-1}. Satellite networks must ensure that this time-averaged EPFD remains below the harmful threshold throughout $\tau$. We provide the complete time-averaged EPFD modeling in \S \ref{sub:interference-model}.

\vspace{-0.1in}
\paragraph{\textbf{2.2.3. Relationship between window size $N_w$, and observation duration $\tau$}} \sysabb is designed to handle LEO and RAS coexistence over the observation duration $\tau$ (e.g., 2000 s, as in~\cite{NRAOODS2025}), whereas the LEO operator automatically reconfigures the satellite network over a scheduling window size $N_w$ (e.g., 15 s for Starlink). Thus, the set of all scheduling windows is defined as $W = \{w_i\}_{i=0}^{|W|-1}$ where $|W| = \tau / N_w$.

\vspace{-0.1in}
\subsection{Telescope Boresight Avoidance (TBA)} \label{sub:tba} 

Unfortunately, today’s LEO satellite operators, such as Starlink, which use TBA methods~\cite{nhan2024spectrum, NRAOODS2025},
do not explicitly take the EPFD constraint into account to protect RAS sites.
TBA relies on reactive, local, rule-based actions without validating the time-averaged interference received at RAS sites and does not ensure that the interference level remains below the permissible threshold. TBA can be outlined as follows:

(1) \textit{Downlink beams are disabled when a satellite passes close to a telescope’s boresight direction}, i.e., the angular separation between the satellite and the telescope boresight positions is less than the ``\textbf{\textit{inner boresight}}'' value (e.g., $\leq 0.5^\circ$ for Green Bank Observatory \cite{nhan2024spectrum}); and 

(2) \textit{Beams of the LEO satellite are steered away from the RAS telescope} (e.g., $180$ km \cite{nhan2025ods}) for satellites within a predefined ``\textbf{\textit{outer boresight}}'' range.

While effective at reducing interference levels, satellites outside the inner and outer boresight continue with normal beam illumination and can collectively generate harmful aggregate sidelobe interference (outlined in \S\ref{sec:interference}), fundamentally limiting the protection level TBA can provide.

\section{\sysabb System}
\label{sec:system}

\sysabb\ addresses these shortcomings by constellation-wide scheduling to protect RAS sites from harmful LEO interference while minimizing user service disruption.

\noindent \textbf{Design Goals:} \sysabb is guided by the following goals:

\begin{figure}[t]
    \centering
    \includegraphics[width=\linewidth]{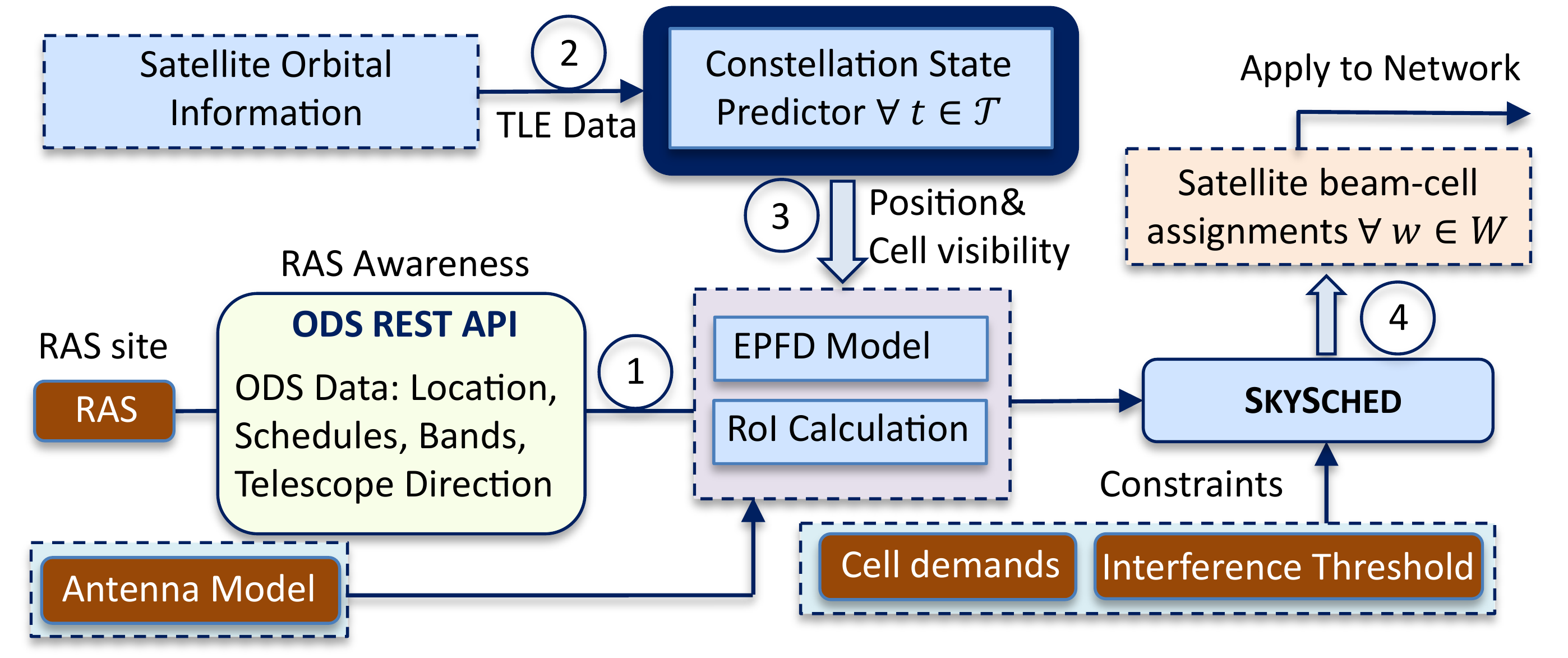}
    \vspace{-0.25in}
    \caption{Overview of \sysabb system.}
    \label{fig:queit_system}
    \vspace{-0.2in}
\end{figure}

\noindent \textbf{1) Strict EPFD compliance.} The primary objective is to ensure that time-averaged EPFD at RAS telescope remains below harmful interference levels.

\noindent \textbf{2) Minimal coverage loss.} \sysabb seeks to preserve coverage near RAS sites while protecting RAS sites, as unnecessary service disruption should be avoided. 

\noindent \textbf{3) Predictive operation.} Because satellite motion and observation schedules are highly predictable, scheduling decisions should proactively exploit future states knowledge.

\noindent \textbf{4) Practical deployability.} The system must operate within existing LEO control-plane infrastructures and not require modifications to satellite hardware or physical-layer beamforming mechanisms.

\vspace{-0.1in}
\subsection{System Architecture}

The \sysabb\ system sits within LEO satellite operator’s cloud control infrastructure and operates over scheduling windows $w$, each of size $N_w$ (e.g., $15$ seconds in case of Starlink Gen2-mini). \sysabb operates as follows (Figure~\ref{fig:queit_system}):

\bstep{1} RAS observatories share their observation metadata, including, telescope location, pointing direction, operating frequency, and observation schedule, through the ODS framework. This information specifies when and where radio astronomy protection is required.

\bstep{2} The \textit{Constellation State Predictor} 
uses the operator’s internal LEO orbital parameters and propagates them to compute precise satellite positions for each satellite over the scheduling interval.

\bstep{3} 
Using predicted satellite states and RAS observation schedules, \sysabb determines satellite visibility for each ground cell, computes relative geometry w.r.t. each RAS telescope, estimates EPFD interference contributions for the entire window size $N_w$, and identifies a Region of Interest (RoI) (detailed in \S\ref{secroi}) around each observatory where satellite beam scheduling decisions materially affect interference.

\bstep{4} At the core of \sysabb is \algname, a predictive scheduler that jointly considers satellite visibility, RAS observation metadata, and interference EPFD thresholds to compute constellation-wide, beam--cell assignments over the scheduling interval. \algname enforces time-averaged EPFD constraints while accounting for cell-level traffic demand, enabling interference-safe coexistence with minimal coverage loss. For EPFD protection, \algname requires only an upper bound on each cell’s demand over the scheduling interval. Such bounds can be estimated from registered users per cell or predicted from recent traffic; using upper bounds is conservative and preserves protection guarantees.

We briefly discuss the practical deployment of \sysabb within the operational constraints of today's LEO satellite networks in \textbf{Appendix~\ref{sec:practicality}}.

\section{Interference Modeling and Protection Constraints}
\label{sec:interference}

This section details how LEO satellite interference impacts RAS, how aggregate interference is modeled using EPFD, and how protection constraints are enforced. We then introduce the EPFD-budgeted Region of Interest (RoI) abstraction, which makes constellation-scale \sysabb optimization computationally tractable. Notation and acronym tables are provided in \textbf{Appendix~\ref{app:notation_table}} and \textbf{Appendix~\ref{app:acronyms}}.

\begin{figure}[t]
    \centering
    \includegraphics[width=1\linewidth]{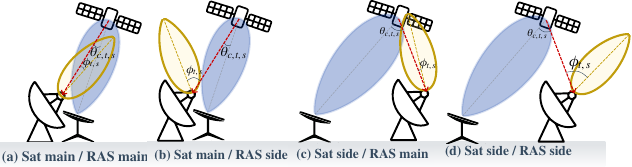}
    \caption{LEO--RAS interference scenario. Blue, yellow, and red beams denote the satellite downlink, the RAS beam, and the interference path from satellite $s$ serving cell $c$ at time $t$, respectively.}
    \label{fig:ra-leo-sharing}
    \vspace{-0.15in}
\end{figure}

\subsection{LEO-RAS Interference Scenarios}
\label{sub:sidelobes}
Although satellite antennas are designed to concentrate energy in narrow main lobes directed toward ground cells, unavoidable sidelobes radiate energy in other directions. Similarly, RAS telescope antennas exhibit strong main lobes toward the observation direction and sidelobes that can receive off-axis signals. Interference arises from interactions between these radiation patterns and constellation geometry.
\sysabb targets in-band coexistence in the 10.7--12.7~GHz downlink range, where RAS observations and satellite downlink transmissions overlap in frequency. RAS may observe in this shared spectrum for continuum and redshifted spectral-line science~\cite{vanZeeAstro2020}; therefore, \sysabb models intended co-frequency downlink emissions, not out-of-band leakage into adjacent protected bands.
As depicted in Figure~\ref{fig:ra-leo-sharing}, four interference scenarios occur:

$1)$ \textbf{\textit{Satellite Main-Lobe to RAS Main-Lobe.}}
A satellite main-lobe $G_{T\text{x},Main}(\theta_{c,t,s})$ aligns with the RAS main-lobe $G_{R\text{x},Main}(\phi_{t,s})$. This produces the strongest interference but is rare, since it requires an alignment of the satellite with the RAS boresight in addition to the satellite serving a cell near the RAS. Figure \ref{fig:ra-leo-sharing}(a) depicts this case.

 $2)$ \textbf{\textit{Satellite Mainlobe to RAS Sidelobe.}}
A satellite serves a nearby cell using its main lobe, but the RAS receives the signal through its sidelobe $G_{R\text{x},Side}(\phi_{t,s})$. This scenario is more common and contributes to moderate interference.

 $3)$ \textbf{\textit{Satellite Sidelobe to RAS Mainlobe.}}
Sidelobe radiation $G_{T\text{x},Side}(\theta_{c,t,s})$ from satellites enters the RAS main lobe. Although individual sidelobe emissions are weak, the RAS’s high main-lobe gain makes this contribution significant.

 $4)$ \textbf{\textit{Satellite Sidelobe to RAS Sidelobe.}}
Both antennas interact through sidelobes. While each contribution is small, 
aggregate interference from many satellites may still contribute non-negligible interference levels. 

As constellation density increases, aggregate interference from sidelobe interactions becomes more important, making RAS protection a fundamentally constellation-wide phenomenon rather than a local geometric event. The existing TBA approach does not properly account for this source of interference and thus fails to protect RAS sites with increased users and satellite density (as shown in \S\ref{sec:evaluation}).

\subsection{EPFD-Based Aggregate Interference}
\label{sub:interference-model}
We follow ITU-R recommendations~\cite{ITUR_S.1586-1} and model interference using EPFD, which is defined as the equivalent single-source power flux density that would produce the same received power in the RAS telescope main beam as the aggregate interference from multiple satellites. 

The instantaneous power flux density (PFD) at the RAS from satellite $s$ transmitting to cell $c$ at time $t$ is given by $\text{PFD}(\theta_{c,t,s}) = \frac{\text{EIRP}(\theta_{c,t,s})}{4\pi D_{t,s}^2}$,
where $D_{t,s}$ is the distance between satellite $s$  at time $t$ and the RAS telescope.  $\text{EIRP}(\theta_{c,t,s})$ is the effective isotropic radiated power toward the RAS at off-axis angle $\theta_{c,t,s}$, we have $\text{EIRP}(\theta_{c,t,s}) = P_{T\text{x},c,t,s} G_{T\text{x}}(\theta_{c,t,s})$, with $P_{T\text{x},c,t,s}$ denoting the transmitted power from satellite $s$ towards cell $c$ at time interval $t$.
Let $\text{e}_{c,t,s}$ denote the EPFD contribution at the RAS telescope from a single spot beam of satellite $s$ serving cell $c$ at any time interval $t \in \T$, we have: $\text{e}_{c,t,s} =
P_{T\text{x},c,t,s} \frac{G_{T\text{x}}(\theta_{c,t,s})}{4\pi D_{t,s}^2} \frac{G_{R\text{x}}(\phi_{t,s})}{G_{R\text{x},max}},$
where $G_{T\text{x}}(\theta_{c,t,s})$ and $G_{R\text{x}}(\phi_{t,s})$ represent the antenna gain of the satellite transmitter and the radio telescope receiver at the given off-axis angles $\theta_{c,t,s}$ and $\phi_{t,s}$, and $G_{R\text{x}, max}$ is the maximum RAS gain. Note that $D_{t,s}$, $\theta_{c,t,s}$, and $\phi_{t,s}$ are received from Constellation State Predictor (See \S\ref{sec:system}). The instantaneous EPFD [W/m$^2$] at time interval $t$ is the 
incoherent linear sum of received flux-density contributions from all active spot beams, following ITU-R S.1586-1~\cite{ITUR_S.1586-1}:

\begin{equation}
\label{eq:1-epfd}
\text{EPFD}(t) = \sum_{s \in S} \sum_{c \in C^{all}} x_{c,t,s} \, \text{e}_{c,t,s},
\end{equation}
where $x_{c,t,s} \in [0,1]$ indicates the normalized capacity of the spot beam from satellite $s$ dedicated to serving cell $c$ at time interval $t$, since we allow time-sharing of a spot beam in any time interval between more than 1 cells.

Figure~\ref{fig:interference-level} shows the interference heatmap as a function of satellite and RAS off-axis angles, illustrating the four interference regimes and highlighting the importance of sidelobe contributions.

\vspace{-0.1in}
\subsection{Time-Averaged Protection Constraints}
\label{sub:epfd-constraint} 
ITU-R RA.1513-2, ITU-R S.1586-1, and ITU-R RA.769-2~\cite{ITU-R_RA.1513-2, ITUR_S.1586-1, ITU-RA769-2} establish that RAS protection is achieved by requiring the time-averaged EPFD over the observation duration to be below the corresponding detrimental threshold, which we denote by $TH_{RAS}$ [W/m$^2$/Hz].
\begin{equation}
\label{eq:epfd_average}
\sum_{t}^{|\mathcal{T}|-1} \text{EPFD}(t) \Delta t \le \tau \, \Delta F \, TH_{RAS},
\end{equation}
where $\tau$ is the total observation duration, $\Delta t$ is the discretization interval, and $\Delta F$ is the observation bandwidth.

The detrimental threshold $TH_{RAS}$ in \eqref{eq:epfd_average} is selected to preserve a specified observation sensitivity and is defined as the interference level that increases the measurement uncertainty by a fixed fraction relative to the receiver-noise: Thermal noise averages down with integration time ($\propto 1/\sqrt{\Delta F\,\tau}$) \cite{ITU-RA769-2}, whereas general interference contributes additional received power that might not vanish under averaging. Thus, if the time-averaged EPFD exceeds $TH_{RAS}$, the observation cannot achieve the intended sensitivity without explicit interference mitigation, and we interpret such violations as RAS data loss. The selection of $TH_{RAS}$ therefore depends on observation type (continuum vs.\ spectral line), bandwidth, the required detection performance (e.g., target SNR), and source visibility, and could be provided by RAS sites or derived from ITU-R specifications \footnote{In co-frequency Ku-band LEO-RAS coexistence, the RA.769-2 detrimental thresholds can be extremely stringent, and often only less sensitivity levels could be achieved.}\cite{ITU-RA769-2}. 
This time-averaged constraint couples scheduling decisions across time and satellites, making radio astronomy protection inherently a multi-satellite, multi-interval optimization problem.

\begin{figure}[t]
    \centering
    \includegraphics[width=0.95\linewidth]{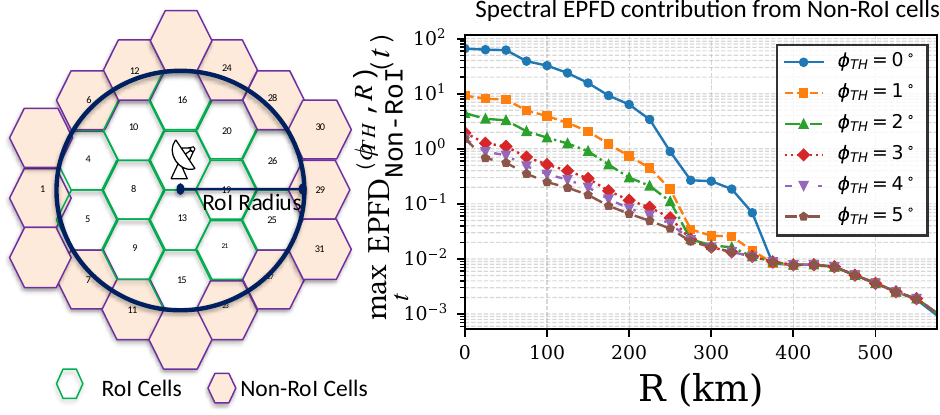}
    \caption{Aggregate spectral EPFD from serving non-RoI cells vs. $R$ and $\phi_{TH}$ when satellites within \textit{${\phi}_{TH}$} of boresight are excluded.}
    \label{fig:roi}
    \vspace{-0.1 in}
\end{figure}

\subsection{EPFD-Budgeted Region of Interest (RoI)}
\label{secroi}

A central challenge in \sysabb is \textit{\textbf{scale}}: In principle, EPFD interference aware scheduling would require jointly optimizing beam assignments for \emph{all} ground cells visible to any satellite whose transmissions may contribute interference at an RAS site. At the constellation scale, such global optimization is computationally infeasible and operationally undesirable, as localized RAS activity could trigger cascading scheduling changes across large portions of the network. Conversely, restricting optimization to a purely geographic neighborhood of the RAS is unsafe. Although distant cells individually contribute little interference, their aggregate sidelobe emissions can still violate EPFD limits~\footnote{This is a key limitation of currently deployed reactive, threshold-based beam avoidance algorithms, such as, Starlink's TBA (\S \ref{sec:intro}).}. 

\sysabb resolves this tension by introducing an \emph{EPFD-Budgeted Region of Interest (RoI)}. Rather than optimizing the entire constellation, \sysabb allocates a residual EPFD budget $\Delta TH_{RAS}$ to non-RoI cells, i.e., cells outside the RoI, and requires that, under the operator’s default scheduling policy, their aggregate interference does not exceed this budget. Only cells whose assignments may jointly exceed $\Delta TH_{RAS}$ are included in \sysabb{}’s interference-aware scheduling optimization. This creates a principled decomposition of the potential serving cells into two regions: (i) a compact set of cells near the RAS handled by \sysabb, and (ii) the remaining cells, which continue operating under standard policies. \sysabb then enforces the effective constraint $\overline{TH}_{RAS} = TH_{RAS} - \Delta TH_{RAS}$, ensuring that optimized cells respect the reduced EPFD threshold while preserving headroom for residual interference. This abstraction bounds worst-case EPFD without requiring joint scheduling of all ground cells,
making \sysabb both safe and scalable across all RAS sites and LEO constellations globally. Importantly, it allows network operators to maintain standard scheduling policies outside the bounded region of interest (RoI), preventing cascading scheduling effects. At the same time, it confines \sysabb\hspace{-0.3em}’s interference-aware scheduling pipeline, \algname (detailed in \S\ref{sec:algorithm}), to the minimal scope necessary for efficient LEO--RAS coexistence.

\textit{\textbf{Determining the RoI and residual EPFD budget.}}
The residual budget $\Delta TH_{RAS}$ directly controls RoI size: smaller values strengthen interference guarantees but increase optimization scope, while larger values reduce computational cost at the risk of higher sidelobe aggregation.

\sysabb Rather than treating $\Delta TH_{RAS}$ as a free parameter, derives it via an online procedure for each RAS site and constellation configuration at each scheduling window.

This procedure exploits two dominant geometric factors affecting EPFD: (i) \textit{the geographic distance $R$} of served cells from RAS site, and (ii) \textit{the satellite--RAS boresight angular separation $\phi_{t,s}$}.
Hence, \sysabb parametrizes the RoI using a geographic radius $R$ (which divides cells into RoI and non-RoI based on the distance of the cell from RAS), and an angular threshold $\phi_{TH}$; to identify satellites that are too close to the RAS boresight to be safely handled by the default, non-RoI scheduling policy (See Eq.~(\ref{eq:cell-roi-def})).

For each candidate pair $(\phi_{TH}, R)$, \sysabb executes its \algname scheduling algorithm over the observation duration $\tau$
and records the resulting EPFD contributions from RoI and non-RoI cells. For each configuration, we compute the worst-case available slack relative to the regulatory limit for $\T$, the set of time intervals that span $\tau$:
\begin{equation}
\text{slack} = TH_{RAS} - \frac{1}{\Delta F}\max_{t \in \T} \Bigl(\text{EPFD}_{\text{RoI}}^{(\phi_{TH}, R)}(t) + \text{EPFD}_{\text{Non-RoI}}^{(\phi_{TH}, R)}(t)\Bigr),
\end{equation}

where $\Delta F$ is the observation bandwidth (See \S\ref{sub:ods-and-timing}). Configurations with positive slack satisfy EPFD constraints while preserving full service feasibility. For such configurations, the implied residual budget is:
\vspace{-0.05in}
\begin{equation}
\Delta TH_{RAS} = \frac{1}{\Delta F}\max_t \text{EPFD}_{\text{Non-RoI}}^{(\phi_{TH},R)}(t).
\end{equation}

\sysabb selects the smallest $(\phi_{TH}, R)$ yielding positive slack, thereby identifying the minimum RoI whose complement provably fits within $\Delta TH_{RAS}$.
The feasible range of $\phi_{TH}$ is determined using the ITU reference RAS antenna model \cite{ITUR-RA.1631-0}, i.e., we sweep $\phi_{TH}$ from $0^\circ$ up to $\phi_{TH}^{\max}=13.6^\circ$, beyond which the gain of the reference RAS antenna drops below $0$~dB.
Range of $R$ is from $0$ to $R_{\max}$, where $R_{\max}$ is the maximum distance from the RAS to any ground cell that can be served by a satellite that remains visible to the RAS (e.g., with a $10^\circ$ RAS elevation, $R_{\max}\approx 3000$~km for a satellite at $600$~km altitude with a $60^\circ$ maximum steering limit \cite{FCC2148}). The RoI calculation algorithm is provided in \textbf{Appendix~\ref{app:roi}.}
Figure~\ref{fig:roi}~(right) illustrates an example of RoI calculation procedure, and plots $\max_t \text{EPFD}_{\text{Non-RoI}}^{(\phi_{TH},R)}(t)$ for values of $R\in (0, 550]$~km, and $\phi_{TH}$ between $0^\circ$ and $5^\circ$. In this example $TH_{RAS}$ was set to $0.4$~Jansky (Jy), and we have positive slack with $\phi_{TH}=5^\circ$, $\Delta TH_{RAS}=0.01$~Jy, and $R=360$~km.

\section{LEO--RAS Coexistence Problem Formulation}
\label{sec:formulation}

We formalize the LEO--RAS coexistence as a constellation-wide scheduling problem. Our goal is to maximize the number of ground cells whose traffic demand is served during a radio astronomy observation, subject to constraints on the time-averaged EPFD at the RAS telescope.

\noindent $\bullet$ \textit{Cells:} Based on the EPFD-budgeted RoI discussion in \S\ref{secroi}, the set of cells considered by \sysabb for satellite $s$ at time $t$, denoted by $C_s^t$, is defined as follows:
\begin{equation} \label{eq:cell-roi-def} 
    \begin{aligned}
        C_{\text{RoI}} &= \{ c \in C^{all} : \text{distance}(c,\text{RAS}) \le R \}, \\ C_s^t&=\{c\in C^{all}: c \text{ visible to } s \text{ at } t,(c\in C_{\text{RoI}} \lor \phi_{t,s}<\phi_{TH})\}\\
    \end{aligned} 
\end{equation}

\noindent $\bullet$ \textit{Cell Demand:}
Each cell $c$ has normalized traffic demand $x_d^c \in (0,1]$, representing the fraction of full spot beam capacity required on average over the observation duration. A value of $x_d^c = 1$ indicates that the cell requires continuous service by one full spot beam, while $x_d^c = 0$ indicates no demand and the cell is not considered in the problem. For instance, a user located in a cell within the RoI subscribes to Residential Lite (FIXED) service plans, for which Starlink provides download speeds of 80–200 Mbps (140 Mbps on average) \cite{service_plans}. With the achievable data rate of each spot beam approximated at 1200 Mbps (see \textbf{Appendix~\ref{app:spot beam-cap}}), the normalized capacity demand $x_d^c$ is $140/1200 = 0.1167$. This example illustrates one user's contribution to the cell load. In general, $x_d^c$ represents the aggregate normalized demand of all users in cell $c$.

\noindent $\bullet$ \textit{Decision Variables:}
We define the following variables: (1) $x_{c,t,s} \in \{0,x_d^c\}$ equals $x_d^c$ if satellite $s$ serves cell $c$ using the required fraction ($x_d^c$) of one spot beam during time interval $t$; and (2) $z_{c,w,s} \in \{0,1\}$ equals 1 if cell $c$ is assigned to satellite $s$ for all intervals in window $w$.

\noindent $\bullet$ \textit{Objective Function:}
We maximize the total number of cells (i.e., all cells within the RoI) whose demand is fully satisfied across the observation duration, i.e., $\sum_{t \in \T,s \in S, c \in \C_s^t}  x_{c,t,s}/x_d^c$.

\noindent $\bullet$ \textit{RAS Protection Constraint:} $e_{c,t,s}$ denotes the EPFD contribution at the RAS from satellite $s$ serving cell $c$ at time interval $t$ (\S\ref{sec:interference}). The time-averaged EPFD protection constraint is:
\begin{equation}
\sum_{t \in \mathcal{T}} \sum_{s \in S} \sum_{c \in C_s^t}
e_{c,t,s} \, x_{c,t,s} \, \Delta t
\le
\tau \, \Delta F \, \overline{TH}_{RAS},
\label{eq:epfd_constraint}
\end{equation}

\noindent $\bullet$ \textit{Satellite Spot Beam Constraints:} Each satellite can transmit at most $N^{sb}$ spot beams simultaneously; Hence its total normalized capacity is limited to $N^{sb}$:
\begin{equation}
\sum_{c \in C_s^t} x_{c,t,s} \le N^{sb}, \quad \forall s, t.
\label{eq:spot beam_limit}
\end{equation}

\noindent $\bullet$ \textit{Single-Satellite Service:} Each cell can be served by at most one satellite at any time:
\begin{equation}
\sum_{s \in S} x_{c,t,s} \le x_d^c, \quad \forall c, t.
\label{eq:single_sat}
\end{equation}

\noindent $\bullet$ \textit{Spot Beam--Cell Consistency:}
To avoid excessive handovers and ensure beam stability, each cell may be served by at most one satellite during each scheduling window:
\begin{align}
& x_{c,t,s} \le x_d^c z_{c,\lfloor t/N_w \rfloor,s}, \quad \forall c,t,s, \label{eq:link_xz} \\
& \sum_{s \in S} z_{c,w,s} \le 1, \quad \forall c,w. \label{eq:unique_z}
\end{align}

\noindent $\bullet$ \textit{Complete Optimization Problem:}
The constellation-wide scheduling problem is a mixed-integer linear program:
\begin{align}
\textbf{P0: } \max_{\{x,z\}} \quad & \sum_{t \in \T} \sum_{s \in S} \sum_{c \in \C_s^t}  x_{c,t,s}/x_d^c \nonumber\\
\text{s.t.} \quad
& \text{constraints}\  \text{(\ref{eq:epfd_constraint})--(\ref{eq:unique_z})} \nonumber \\
&z_{c,w,s}\in\{0,1\},x_{c,t,s}\in\{0,x_d^c\} \nonumber 
\end{align}

\begin{theorem}
\label{theorem:nphard}
Problem \textbf{P0} is NP-hard
(proof in Appendix~\ref{app:nphard}).
\end{theorem}


\section{\algname Algorithm}
\label{sec:algorithm}

This section presents \textsc{SkySched}, a centralized, polynomial-time, flow-based scheduling algorithm to solve the LEO--RAS coexistence problem formulated in \S\ref{sec:formulation}. The key idea is to convert satellite--cell assignment decisions into min-cost, flow selection over a time-expanded flow graph, where EPFD interference is modeled as edge cost and spot beam availability as capacity constraints. We first show that for a restricted special case, the formulation admits an exact solution via minimum-cost flow. We then extend this approach to the general NP-hard case using iterative feasibility repair.

\begin{figure}[t]
    \centering

    \begin{subfigure}[t]{0.5\linewidth}
        \includegraphics[width=\linewidth]{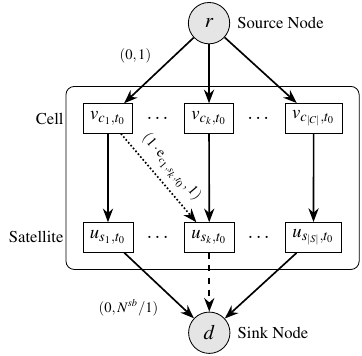}
        \caption{Flow network $G^f$ for a single time interval $t_0$}
        \label{fig:graph_left}
        \vspace{-0.06em}
    \end{subfigure}
    \begin{subfigure}[t]{0.44\linewidth}
        \includegraphics[width=\linewidth]{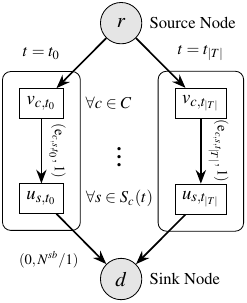}
        \caption{Time-expanded $G^f$ for scheduling interval.}
        \label{fig:graph_right}
        \vspace{-0.06em}
    \end{subfigure}
    \caption{Flow network used by \textsc{SkySched}. Each layer corresponds to a scheduling time interval (or window when $N_w>1$). Flow from source $r$ through $(c,t)$ and $(s,t)$ to sink $d$ represents assigning satellite $s$ to serve cell $c$ during interval $t$. Arc costs capture EPFD contribution; Arc capacities enforce spot beam limits.}
    \vspace{-0.2in}
    \label{fig:flow_graph}
\end{figure}

\subsection{Special Case: Unit Window and Uniform Demand}
\label{sec:specialcase}

We first consider the case where $N_w=1$, $x_d^c = 1$. The solution for this case corresponds to: i) letting the satellite-cell association change at every time interval; ii) every cell having the same demand, equal to the full capacity of one spot beam; and iii) the minimum service accepted at the cell is equal to the cell demand of 1, which implies we either can serve the cell with a full spot beam, or not serve it at all ($x_{c,t,s} \in \{0,1\}$).

Next, we show that in this special case, \textbf{P0} admits an exact reduction to a series of interference minimization subproblems, each solvable via a min-cost flow formulation.

\subsubsection{Interference Minimization Subproblem} \label{sub:in-min-sub} 

Let $\textbf{P}_1(q)$ be the problem of finding the minimum-interference assignment $x_{c,t,s}$ such that $\sum_{c,t,s} x_{c,t,s}/x_d^c = q$. Notice that since $N_w = 1$, constraints in (\ref{eq:link_xz})-(\ref{eq:unique_z}) are simplified to: $\sum_{s \in S} x_{c,t,s} \leq 1$, and $x_{c,t,s} \in \{0,1\}$. \textbf{Appendix~\ref{app:p1q}} gives the formal definition of the corresponding optimization problem \(\textbf{P}_1(q)\). Let us denote the optimal value of $\textbf{P}_1(q)$ as $\text{EPFD}^\star(q)$. In \text{\large L}\text{\small EMMA} \ref{lemma:monotone} and \text{\large T}\text{\small HEOREM} \ref{theorem:q-star-optimal}, we provide proof that the optimal solution of $\textbf{P}_0$ is obtained by solving $\textbf{P}_1(q)$ at largest integer $q$ that yields $\text{EPFD}^\star(q) < \Delta F\frac{\tau}{\Delta t} \overline{TH}_{RAS}$.

\begin{lemma}
\label{lemma:monotone}
$\text{EPFD}^\star(q)$ is non-decreasing in $q$.
\end{lemma}
\noindent\emph{Proof sketch.} We prove by contradiction: if $q<q'$ yet we have $\text{EPFD}^\star(q')<\text{EPFD}^\star(q)$, then dropping $(q'-q)$ active elements from an optimal $q'$-solution yields a feasible $q$-solution with no larger EPFD since interference is non-negative, a contradiction. See details in \textbf{Appendix~\ref{app:non-decreasing}}.

\begin{theorem} \label{theorem:q-star-optimal}
Let
\[
q^\star = \max \{ q \mid \text{EPFD}^\star(q) < \tau/\Delta t \Delta F \overline{TH}_{RAS} \}.
\]
Then $q^\star$ is the optimal objective value of \textbf{P0} under the special-case assumptions, and any assignment that achieves $q^\star$ is an optimal assignment for our original problem \textbf{P0}
\end{theorem}

\noindent\emph{Proof sketch.} Feasibility follows because any solution to $\mathbf{P}1(q)$ with $\text{EPFD}^\star(q) < \tau/\Delta t \Delta F\overline{TH}_{RAS}$ satisfies all constraints of \textbf{P0}. Optimality follows from Lemma~\ref{lemma:monotone}. See the details in \textbf{Appendix~\ref{app:theorem6-2}}.

\subsubsection{Flow Network Construction}
\label{sub:in:min-gen}
To solve $\mathbf{P}_1(q)$ subproblem, we construct a time-expanded flow network $G^f=(N,A)$ that encodes all feasible satellite--cell assignments over the scheduling interval.

\textbf{Nodes.}
For each time step $t\in\mathcal{T}$, we create a set of cell nodes $\{v_{c,t}\}$, one for each cell $c$ within the region of interest at time interval $t$, and a set of satellite nodes $\{u_{s,t}\}$, one for each satellite $s$ in the constellation. We further introduce a source node $r$ and a sink node $d$.

\textbf{Assignment edges.}
For every visible satellite--cell pair at time interval $t$, i.e., when satellite $s$ is visible from cell $c$, we add a directed edge $(v_{c,t},u_{s,t})$. Selecting this edge corresponds to assigning satellite $s$ to serve cell $c$ at time interval $t$. Each such edge has unit capacity and incurs a cost equal to the EPFD contribution $e_{c,t,s}$.

\textbf{Source and sink edges.}
To enforce per-cell and per-satellite constraints, we add edges $(r,v_{c,t})$ with capacity $1$ and zero cost, ensuring that each cell can be served by at most one spot beam at time interval $t$. Similarly, we add edges $(u_{s,t},d)$ with capacity $N^{sb}$ and zero cost, limiting the number of concurrent spot beams on the satellite at each time interval to $N^{sb}$ (Figure~\ref{fig:flow_graph}).

\textbf{Flow interpretation.}  A flow $f(r, v_{c,t})$, $f(v_{c,t}, u_{s,t})$, or $f(u_{s,t}, d)$ is a non-negative attribute of an arc in $G^f$ that must be less than the capacity of the arc and corresponds to the spot beam capacity that passes through it, 
The total cost of the flow of value $q$ from $r$ to $d$ in $G^f$ is thus equal to the aggregate EPFD induced by the assignments; consequently, minimizing flow cost yields the minimum-interference assignment for the given flow $q$ (see Theorem~\ref{theorem:equiv}).

\vspace{-0.05in}
\begin{lemma}\label{lemma:one-to-one}
For a given integer flow demand $q$, there is a one–to–one correspondence between feasible solutions $x_{c,t,s}$ of Problem $\textbf{P}_1(q)$, and \(q\)-unit, integral flows \(f\) in $G^f$.
\end{lemma}

\noindent\emph{Proof Sketch.} The result follows from the one-to-one mapping between binary assignments $x_{c,t,s}\in \{0, 1\}$ and integral $q$-unit flows induced by the network construction, with arc capacities enforcing the constraints (see \textbf{Appendix~\ref{app:one-to-one}}).

\begin{theorem}
\label{theorem:equiv}
For fixed $q$, the minimum-cost $q$-unit flow on $G^f$ yields exactly $\text{EPFD}^\star(q)$.
\end{theorem}

\noindent\emph{Proof sketch.} From Lemma~\ref{lemma:one-to-one} we know there is one-to-one correspondence between integral solutions; and from integrality theorem \cite{AhujaMagnantiOrlin1993-Theorem9.10}, Since $G^f$ has integral capacities and demand, the min-cost flow problem admits an optimal integral solution (see full proof in \textbf{Appendix~\ref{app:flow_equiv}}).

\subsubsection{Uniform $x_d^c = x_d < 1$}
The min-cost flow formulation extends to any uniform demand $x_d\in(0,1]$ by multiplying each assignment-arc cost by $x_d$
and scaling each satellite node capacity to $\left\lfloor N^{sb}/x_d\right\rfloor$ (one spot beam can be time-shared between multiple low-demand cells). The min-cost flow solution remains optimal.

\subsection{General Case: Large Window Size and Heterogeneous Demand} \label{sub:hetero}

For arbitrary cell demand ($x_d^c$) and window size ($N_w>1$), \textbf{P0} is NP-hard (\textbf{Theorem \ref{theorem:nphard}}) and cannot be solved optimally by min-cost flow, hence, we design a heuristic as follows: 

\noindent $\bullet$ \textit{Window-Based Graph Construction:} We aggregate time slots into windows and construct nodes $(c,w)$ and $(s,w)$. An assignment persists throughout window $w$.

\noindent $\bullet$ \textit{Flow Splitting under Heterogeneous Demand:}
We allow each cell-window to contribute at most one unit of flow, that is: $\texttt{cap}(r,v_{c,w})=1$. One unit of flow corresponds to selecting exactly one serving satellite for $(c,w)$, and we encode heterogeneous demand only in the arc costs, by setting cost of $v_{c,w}$ to $u_{s,w}$ be $\sum_{t\in w} x_d^c\,e_{c,t,s}$.
This prevents flow splitting such that a single $(c,w)$ does not send flow to multiple satellites (which would violate $\sum_{s\in S}z_{c,w,s} \leq 1$).

\noindent $\bullet$ \textit{Scaled Satellite Spot beam Capacity:}
To approximate the per-satellite spot beam budget $\sum_c x_d^c\,x_{c,w,s}\le N^{sb}$ within a unit-flow, heterogeneous demand model, we treat one selected $(c,w)$ as consuming $\overline{x_d}$ units of capacity where $\overline{x_d}$ denotes the average demand of cells in the RoI, and set the capacity of edges from $u_{s,w}$ to $d$ to be $\Big\lfloor \frac{N^{sb}}{\overline{x_d}} \Big\rfloor$.
Thus the flow chooses a non-splitting set of assignments, while the satellite-side capacity is an average-demand proxy that we later repair via feasibility repair.

\noindent $\bullet$ \textit{Feasibility Repair:} After each min-cost flow solution, we compute true utilization $\rho_{s,w} = \sum_{c: x_{c,w,s}=1} x_d^c$.
Following that, if we have $\rho_{s,w} > N^{sb}$, we remove assignments with lowest EPFD until feasibility is restored.
If $\rho_{s,w} \ll N^{sb}$, we increase $\texttt{cap}(u_{s,w},d)$ and re-solve the flow to reclaim unused capacity.
This solve--repair--update process is repeated until convergence or iteration limit.

\subsection{\textsc{SkySched} Algorithm Overview}
\label{sub:algoverview}

Here, we provide an overview of the proposed \algname algorithm for Special Case \S\ref{sub:in-min-sub} and General case \S\ref{sub:hetero}. For the pseudo-code details, refer to \textbf{Appendix~\ref{app:alg}}.

The \textsc{SkySched} algorithm first determines satellite locations and orbits from TLE data (e.g., CelesTrak \cite{Celestrak}) at each time $t\in\T$, and uses the radio observatory specification from ODS to determine the geometric relationships among satellites, cells, and the RAS. For each visible satellite--cell pair in the RoI, it computes the EPFD contribution $\text{e}_{c,t,s}$ at the RAS, accounting for antenna patterns and off-axis angles $\{ \theta_{c,t,s}, \phi_{t,s}\}$. \textsc{SkySched} then solves for the maximum served cell--time pairs under the EPFD budget by performing a binary search over a flow demand $q$, and repeatedly solving min-cost flow instances on the Flow Network of Figure~\ref{fig:flow_graph}. In the special case ($N_w=1$ and uniform demand) it returns the largest $q^\star$ such that $\text{EPFD}^\star(q^\star)\le \tau/\Delta t\Delta F\overline{TH}_{RAS}$. In the general case with larger window size and heterogeneous demands ${x_d^c}$, it solves a scaled min-cost flow to propose assignments, then applies a feasibility-repair step that prunes overloaded $(t,s)$ pairs to satisfy the true spot-beam capacity $\sum_c x_d^c \le N^{sb}$, iterating
until convergence or an iteration cap, and returns the best feasible assignment found.

\textit{\textbf{Time Complexity Analysis.}} Let $|\overline{S}|$ be the average number of satellites visible per cell. We have $|N|=O(|C_{\text{RoI}}||W|+|\overline{S}||W|)$, $|A|=O(|C_{\text{RoI}}||\overline{S}||W|)$. From \cite{GoogleORToolsMinCostFlow}, each solve of min-cost flow using cost-scaling algorithm has complexity of $O\!\left(|N|^{2}\,|A|\,\log\!\big(|N|\cdot|\text{EPFD}|\big)\big)\right)$, where $|\text{EPFD}|$ is the largest arc cost. Binary search requires $O(\log(|C_{\text{RoI}}||W|))$ solves, hence, we can see that complexity of \algname\ is: $O\!\left(|N|^{2}\,|A|\,\log\!\big(|N|\cdot|\text{EPFD}|\big)\,\log\!\big(|C_{\text{RoI}}||W|\big)\right)$;

\begin{figure*}[t]
    \centering
    \includegraphics[width=0.92\linewidth]{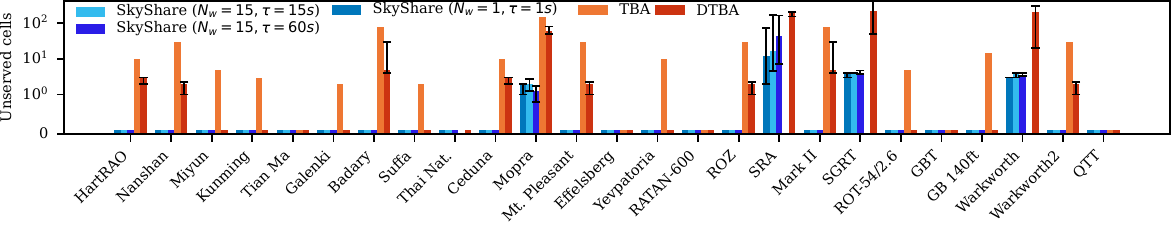}
    \vspace{-0.15in}
    \caption{Number of  cells unserved (Y-axis) by \sysabb, TBA and DTBA at \NSites RAS  worldwide (x-axis).}
    \label{fig:global_ras}
    \vspace{-0.15in}
\end{figure*}

\section{Evaluation}
\label{sec:evaluation}

This section presents the evaluation of our \sysabb system and compares it to SoTA LEO-RAS coexistence methods.

\vspace{-0.1in}
\subsection{Experimental Setup}

\textbf{LEO satellite constellation.} In our experiments, we use Supplemental General Perturbation-derived TLEs for the Starlink Gen2-mini constellation obtained from CelesTrak (detailed in \textbf{Appendix~\ref{app:leo_supgp_background}}). We simulate constellation dynamics over a 24-hour period from 2026:01:22T08:00:00 to 2026:01:23T08:00:00 UTC using an extended SpaceNet simulator \cite{spacenet}, derived from \textit{x}eoverse~\cite{xeoverse}. Experiments across additional time ranges obtained consistent results that we omit due to space constraints.

\noindent
\textbf{Antenna Model}. For RAS, we use the reference antenna from ITUR-RA.1631~\cite{ITUR-RA.1631-0}. 
Starlink antenna is modeled as a 2D, $25\times40$ half-wavelength uniform rectangular phased array (See \ref{app:star-antenna}).

\noindent
\textbf{Cell Model.}
The quasi-Earth-fixed cells are constructed using an H3 hexagonal grid at resolution 5~\cite{h3} that mimics the Starlink spot beam footprint ~\cite{sorensen2024fixed}. At this resolution, each cell has an average area of 252.9 $\text{km}^2$
~\cite{h3-measure}.

\noindent
\textbf{Observation Metadata.}
We focus on challenging intervals by sampling time ranges that include at least one satellite passing close to the boresight (angular separation below $1^\circ$). We model an RAS conducting a series of observations whose schedule is provided using ODS. The observation declination is set to the RAS latitude, which guarantees that during the full day simulation, there exists a time when the observation direction points directly overhead of the RAS.
We sample $10$~minute passages for each latitude at times when there are passages of satellites close to the RAS boresight. 

Unless stated otherwise, we use Ku-band downlink (10.7--12.7~GHz), Interference threshold of $0.4$~Jy (based on reported TBA simulations from \cite{NRAOODS2025}), 
$x_d^c=1$, so that every served cell requires one full spot beam, providing a conservative full-load setting, window size $N_w=\{1,15\}$~s, observation duration ($\tau$) equal to 1-second, 15-second, and 60-seconds, 
Other network parameters are summarized in \textbf{Appendix~\ref{app:table}}.

\vspace{-0.1in}
\subsection{Metrics and Baselines}

\textbf{Metrics.}
We report:
(i) average number of unserved cells during observation time, showing cell coverage loss due to coexistence,
(ii) time-averaged EPFD, the ITU-defined protection metric,
(iii) achieved RAS sensitivity, to signify the scientific impact of coexistence on observation,
(iv) algorithm scalability and optimality, and 
(v) robustness to system errors. We use $U$ to denote the number of unserved cells (over the observation duration). Thus, $U=K$ denotes operation when exactly $K$ cells are unserved.

\noindent \textbf{Baselines.}
We compare \sysabb against: 

$\bullet$ \textbf{TBA}: fixed-angle beam avoidance near RAS boresight with parameters selected such that worst-case scenarios would not result in harmful interference levels (See \S\ref{sub:tba}).

$\bullet$ \textbf{Dynamic TBA (DTBA)}: We adopt DTBA from TBA, assuming that at each time, parameters (inner, and outer angular threshold, exclusion zone, see \S \ref{sub:tba}) can be changed to achieve the best cell coverage. This method is not currently adopted by Starlink, which uses TBA \cite{NRAOODS2025}.

\vspace{-0.1in}
\subsection{Experimental Results}

\noindent{\textbf{Global RAS Sites.}}
To assess global impact, we consider all single-dish RAS sites whose observing frequencies directly overlap with Starlink downlink. This includes \NSites~out of approximately $50$ operational Ku-band RAS sites worldwide.
The remaining Ku-band RAS sites are interferometric arrays, which are generally more tolerant to satellite interference when operating in array mode~\cite{ITU-RA769-2} and when elements are widely separated~\cite{van2011radio}. Although arrays can perform single-dish observations~\cite{van2011radio,ITU-RA769-2}, we choose to focus on the most interference-sensitive class of RAS.

Across the considered $25$ Ku-band RAS sites worldwide, Figure~\ref{fig:global_ras} shows that TBA and DTBA leave large geographic regions underserved in order to protect RAS operation; in contrast, \sysabb reduces total unserved cells by over \textbf{90.68\%} on average. This recovery enables LEO constellations to serve substantially more users globally, particularly in underserved regions, while maintaining strict protection for scientific observation sites.

\begin{figure}
    \centering
    \includegraphics[width=1\linewidth]{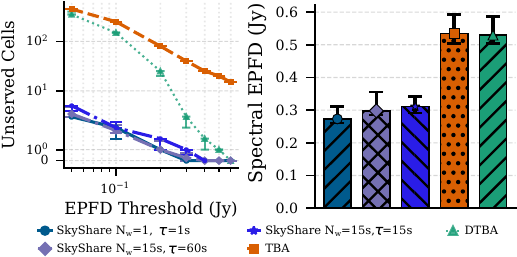}
      \textbf{(a) Unserved vs threshold.\hfill
      (b) Mean EPFD for U=0.}
      \vspace{-0.1in}
    \caption{\sysabb vs. Baselines: (a) Unserved vs Threshold. (b) Mean spectral EPFD with U = 0.}
        \vspace{-0.1in}
    \label{fig:exp1}
\end{figure}

\noindent{\textbf{Coverage and EPFD Interference Tradeoff.}}
We evaluate a representative $30$~m RAS telescope at latitude $40^\circ$, reflecting the typical deployment of Ku-band RAS sites in both antenna size and location (see \S\ref{sec:ras_telescope_operation}, Figures~\ref{fig:antenna_hist} and \ref{fig:ras_latmpact}). Figure~\ref{fig:exp1}(a) shows that \sysabb maintains near-zero unserved cells even under stringent EPFD thresholds, whereas TBA rapidly sacrifices coverage as interference constraints tighten. Increasing TBA inner/outer avoidance angles cannot reduce EPFD without substantial service loss, since doing so excludes satellites that would otherwise serve RoI cells. In contrast, \sysabb exploits constellation-wide diversity to preserve coverage while meeting interference limits. Figure~\ref{fig:exp1}(b) shows the mean for EPFD when forcing TBA, DTBA, and \sysabb to have zero cell loss ($U=0$). For $N_w=1$~s and $N_w=15$~s, \sysabb reduces mean EPFD by more than 48\% and 41\%, respectively, showing the advantage of constellation-level coordination. 

\begin{figure}[t]
    \centering
    \begin{subfigure}[t]{0.59\linewidth}
        \centering
        \includegraphics[width=\linewidth]{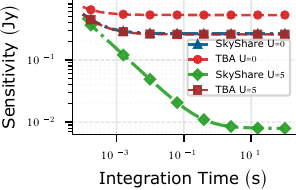}
        \caption{Sensitivity vs Integration time.}
        \label{fig:sens-time}
    \end{subfigure}
    \begin{subfigure}[t]{0.37\linewidth}
        \centering
        \includegraphics[width=\linewidth]{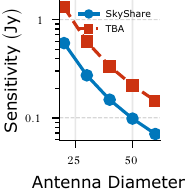}
        \caption{Sensitivity vs. diameter.}
        \label{fig:sens-diameter}
    \end{subfigure}
    \vspace{-0.12in}
    \caption{Impact of residual satellite interference on achievable observation sensitivity.}
    \vspace{-0.16in}
\label{fig:sensitivity_bundle}
\end{figure}

\noindent{\textbf{Scientific Impact: RAS Observation Sensitivity.}}
Residual satellite interference could introduce an integration independent noise floor that limits achievable RAS sensitivity. Figure~\ref{fig:sensitivity_bundle}(a) shows that for a $30$~m telescope under full coverage ($U=0$), \sysabb achieves a sensitivity floor of $0.264$~Jy versus $0.530$~Jy for TBA (a $3$\,dB improvement). Allowing modest coverage loss further widens this gap: at $U=5$, \sysabb provides up to $15$\,dB better sensitivity. Figure~\ref{fig:sensitivity_bundle}(b) shows similar gains across antenna diameters; even at $D=60$~m, \sysabb improves sensitivity by $3.38$\,dB, indicating that constellation-wide coordination directly translates into meaningful improvements in scientific sensitivity.

\label{subsec:impact_ras}

\begin{figure}
    \centering
    \includegraphics[width=1\linewidth]{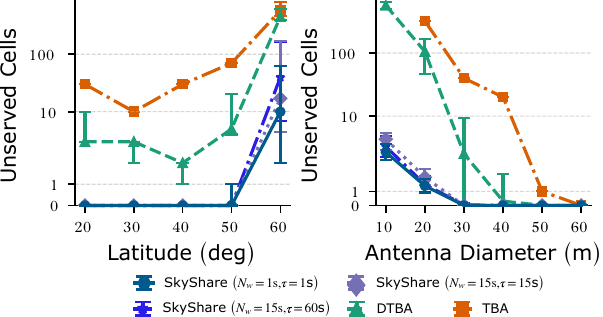}
      \textbf{(a) Unserved vs. RAS lat.\hfill
      (b) Unserved vs diameter.}
      \vspace{-0.12in}
    \caption{\textsc{\sysabb} vs TBA and DTBA baselines. (a) Unserved cells vs RAS latitude. (b) Unserved cells vs RAS antenna diameter.}
    \vspace{-0.2in}
    \label{fig:unserved-vs-d-lat}
\end{figure}

\noindent{\textbf{Geographic and Physical Factors.}}
Figure~\ref{fig:unserved-vs-d-lat} (a) shows that TBA degrades sharply at high latitudes due to reduced satellite diversity. This is specifically relevant for European Ku-band RAS sites above $50^\circ$, e.g., Effelsberg~\cite{effelsberg-mpifr}, Mark II~\cite{markii-jodrell}. At $60^\circ$, \sysabb improves the coverage by $94.85\%$ and $92.72\%$ relative to TBA and DTBA, respectively.

Larger dishes have narrower receive beams and lower susceptibility to interference. Figure~\ref{fig:unserved-vs-d-lat}(b) shows that for RAS diameters above $30$~m, \sysabb achieves full coverage under a $0.4$~Jy EPFD threshold. Even for $10$ m telescopes, \sysabb reduces unserved cells by $99.33\%$ whereas TBA requires large exclusion zones due to increased sidelobe coupling. Together, these results show that \sysabb consistently outperforms reactive TBA and DTBA baselines across diverse geographic locations and telescope characteristics.

\begin{figure}[t]
  \centering
  \includegraphics[width=1\linewidth]{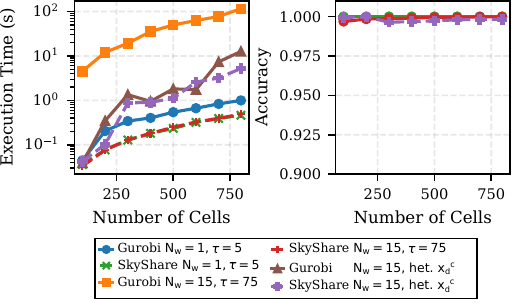}
  \textbf{(a) Exec. Time vs RoI size.\hfill
  (b) Accuracy vs RoI size.}
  \vspace{-0.1in}
  \caption{Scalability and Optimality Analysis.}
  \vspace{-0.18in}
  \label{fig:exp5_runtime_and_accuracy_vs_cells}
\end{figure}

\noindent{\textbf{Scalability and Optimality.}} 
We compare \sysabb with Gurobi~\cite{gurobi}, an exact Mixed Integer Linear Program (MILP) solver. Figure~\ref{fig:exp5_runtime_and_accuracy_vs_cells}(a) shows that Gurobi runtime grows exponentially with RoI size, while \sysabb scales polynomially. Figure~\ref{fig:exp5_runtime_and_accuracy_vs_cells}(b) shows \sysabb achieves over $99\%$ of optimal across tested instances. For the heterogeneous-demand setting, we assign high demand ($x_d^c \in [0.9,1.0]$) to $90\%$ of cells and low demand ($x_d^c \in [0,0.1]$) to the remaining $10\%$, and report averages over five random demand samples.
\algname remains above $99\%$ of the Gurobi optimum in heterogeneous-demand case, while retaining lower runtime at larger RoIs.

\noindent{\textbf{Robustness to System Errors.}} 
We stress-test \sysabb under Gaussian perturbations in ODS pointing direction and satellite position uncertainty. Under pointing error in Figure~\ref{fig:point-error}, \sysabb incurs $0\%$ mean EPFD violations for RMS error $\leq 12$ deg, rising gradually to at most $38\%$ EPFD violations beyond this regime, while outperforming both TBA and DTBA across all error levels.  Figure~\ref{fig:pos-error} shows that \sysabb is sensitive to satellite position error, with mean EPFD violations exceeding $80\%$ at $30$\, km error. By contrast, the TBA baseline is resilient due to its conservative exclusion of $30$ cells around the RAS site, while \sysabb serves all cells.
This exposes a key tradeoff: \sysabb relies on accurate satellite localization, which is feasible in practice given star-tracker–based positioning used by operators such as Starlink~\cite{starlink_stargaze_2024}. Finally, DTBA exhibits intermediate robustness by alternating between $0$ and $30$ unserved cells over the total observation duration, averaging $3.53$ unserved cells across the observation duration.

\begin{figure}[t]
    \centering
    \begin{subfigure}[t]{0.43\linewidth}
        \includegraphics[width=\linewidth]{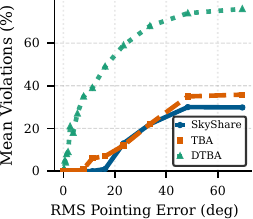}
        \caption{Violations Due to Pointing Mismatch}
        \label{fig:point-error}
    \end{subfigure}
    \hfill
    \begin{subfigure}[t]{0.45\linewidth}
        \centering
        \includegraphics[width=\linewidth]{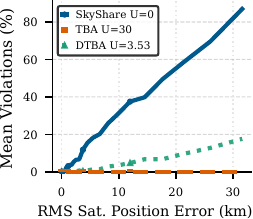}
        \caption{Violations Due to Position Prediction Error}
        \label{fig:pos-error}
    \end{subfigure}
    \vspace{-0.12in}
    \caption{Stress Tests: EPFD Threshold Violations vs (a) RAS Pointing error, and (b) Satellite Position Error}
    \label{fig:combined_3}
    \vspace{-0.2in}
\end{figure}

\section{Discussion and Limitations} \label{sec:discussion}

\noindent{\textbf{Predictive coordination versus reactive avoidance.}}
Our results suggest that LEO--RAS coexistence is inherently a constellation-wide optimization problem rather than a local reactive avoidance task. Existing beam-avoidance methods treat satellites independently and ignore aggregate sidelobe interference, often causing substantial coverage loss without proportional interference reduction. By leveraging predictable satellite motion and scheduled observations, \sysabb enables proactive, network-wide coordination that improves coverage while protecting RAS operations.

\noindent{\textbf{Operational assumptions.}}
\sysabb assumes centralized control-plane scheduling and access to observatory metadata via ODS, consistent with Starlink's global satellite--user assignment architecture~\cite{tanveer2023making}. 
Distributed operation is left as future extension, trading signaling overhead for decentralized control. Assignments remain fixed within each scheduling window; the evaluated $15$-s windows therefore bound reassignment and handover opportunities.

\noindent{\textbf{Single-Beam vs. Multi-Beam Service.}}
In \sysabb, at most one spot beam serves each cell. This model targets RoI cells whose demand does not exceed one spot beam's capacity, a plausible regime in low-density regions around many RAS sites~\cite{ITU-R_RA.2259}.
Multi-beam service would allow high-demand cells to receive concurrent
service from multiple beams or satellites.
To support this in \textbf{P0}, $x_d^c$ would be allowed to exceed one, and $x_{c,t,s}\in[0,x_d^c]$ would represent the total normalized spot-beam capacity allocated by satellite $s$ to cell $c$. Constraint~\eqref{eq:single_sat} would remain as the aggregate-demand bound $\sum_s x_{c,t,s}\leq x_d^c$. Constraint~\eqref{eq:unique_z} would be relaxed by replacing its current upper bound of one with the permitted number of concurrent serving beams or satellites, while Constraint~\eqref{eq:link_xz} would continue to enforce window consistency for every active satellite--cell assignment. This concurrency limit may depend on the available frequency and polarization resources, e.g., earlier SpaceX filings modeled one co-frequency beam per Ku-band spot in their interference analysis~\cite{FCC2148}.
Under this extension, EPFD remains additive because each satellite--cell allocation contributes $e_{c,t,s}x_{c,t,s}$ to Constraint~\eqref{eq:epfd_constraint}.
If frequency and polarization are represented explicitly, the corresponding contribution becomes $e_{c,t,s}^{f,p}x_{c,t,s}^{f,p}$.
However, splitting cell demand across satellites increases scheduling
complexity.
This expands the assignment space and requires a splittable-flow formulation or demand discretization; the feasibility-repair must also preserve each cell's aggregate allocated capacity while resolving satellite-capacity and EPFD violations. Schedules produced under the current model remain feasible in a multi-beam architecture, but do not exploit its larger feasible assignment set.

\noindent{\textbf{Limited Networking Scope.}}
This work focuses on LEO and RAS coexistence with the objective of preserving cell coverage under interference constraints. We do not explicitly model network-level metrics such as latency, queueing, power control dynamics, or inter-cell interference; incorporating such objectives into the scheduler remains future work. Our evaluation also focuses on a single operator (Starlink); multi-operator scenarios remain future work.


\noindent{\textbf{System Model Fidelity.}}
Our study uses ITU-compliant EPFD modeling, public TLE data, and representative antenna patterns. If measured antenna patterns exhibit poorer sidelobe suppression, the per-assignment EPFD terms $e_{c,t,s}$ could increase for some cells. Because \sysabb computes both the optimized RoI contribution and the maximum non-RoI contribution using the specified antenna pattern, these changes would be reflected directly in the RoI calculation. To maintain the same protection threshold, the selected RoI may grow, increasing scheduling complexity and runtime, or a larger non-RoI interference budget $\Delta TH_{RAS}$ may be reserved, reducing the remaining budget $\overline{TH}_{RAS}$ available to RoI scheduling and potentially increasing the number of unserved cells. Conversely, better sidelobe suppression may shrink the RoI and leave a larger EPFD budget for RoI scheduling. While these models capture key effects, validating \sysabb under operational traffic traces and real-world LEO network and antenna measurements remains future work.

\noindent{\textbf{Interference Scope.}}
\sysabb applies to the intended co-frequency transmissions in the 10.7--12.7~GHz range described in \S\ref{sub:sidelobes}. We do not model leakage into adjacent protected bands, harmonic or spurious emissions outside the satellite operating band. Such emissions are not represented by the current $e_{c,t,s}$ coefficients, and their sidelobe patterns may differ substantially from the in-band phased-array pattern. Frequency-dependent emission masks and antenna-pattern measurements would be required to model these effects and construct band-specific coefficients $e_{c,t,s}^{f}$ and enforce separate EPFD constraints for each affected band.

\noindent{\textbf{Multiple simultaneous RAS observations.}}
Our evaluation considers protecting a single RAS observation at a time. In practice, multiple observatories may operate concurrently and point to different regions of the sky, introducing overlapping EPFD constraints. While \sysabb can incorporate multiple constraints in principle, jointly optimizing across many simultaneous observations
for multi-observatory protection is an important direction for future work.


\vspace{-1em}
\section{Related Work}
\label{sec:related}

\noindent{\textbf{LEO Satellite and RAS Coexistence.}}
LEO-RAS interference has been explored recently with efforts from satellite network operators and RAS sites \cite{di2023unintended, grigg2023detection, nhan2024spectrum, nhan2025ods, Starlink}. Notably, intended and unintended LEO emissions have been detected at SKA-Low site in SKA-Low frequency range \cite{grigg2023detection}, with \cite{di2023unintended} detecting unintended emissions in LOFAR RAS. ODS framework \cite{NRAOODS2025} has been proposed as a "self-reporting system" to share observation schedules with satellite network operators,
facilitating coexistence of RAS and satellite networks. Starlink has introduced TBA \cite{Starlink}
to adjust its LEO constellation operation for RAS protection. The ODS-based TBA system has been tested through coordination with NRAO to verify correct operation \cite{nhan2025ods, NRAOODS2025}. While effective, TBA is based on static, rule-based mechanism and could negatively affect network availability, specially for lower frequency bands, smaller RAS antenna diameter size, or wider spot beam footprints, which results in larger angular exclusion, and affects satellites for a longer time~\cite{nhan2025ods}; in contrast, \sysabb enables predictive, constellation-wide coordination that explicitly optimizes interference and service, achieving substantially better coexistence outcomes.

\noindent{\textbf{Dynamic Spectrum Sharing.}}
Dynamic spectrum access has been extensively studied in terrestrial wireless, including cognitive radio for opportunistic access to underutilized spectrum, and Database-driven spectrum access  \cite{haykin2005cognitive, akyildiz2006next, zheleva2023radio, sohul2015spectrum}. However, many sharing frameworks operate with response times on the order of minutes to hours (e.g., TVWS and CBRS) \cite{zheleva2023radio}, reflecting assumptions that are often compatible with relatively slowly varying conditions rather than highly time-varying geometries \cite{zheleva2023radio}. Prior work notes that many spectrum sensing approaches assume stationary sensors, motivating explicit treatment of mobility \cite{min2009impact}. In contrast, \sysabb exploits the predictability of satellite motion and observation schedules to enable proactive, constellation-wide coordination tailored to fast time-varying geometry.

\noindent{\textbf{Satellite Network Control.}}
Prior work on satellite network control has explored routing
\cite{handley2018delay, li2024stable}, gateway selection
\cite{abubakar2024choosing, jang2025geo}, and traffic engineering
\cite{wu2025sate} in rapidly changing LEO constellations. Related
end-to-end control, such as congestion control, also targets performance
under LEO dynamics~\cite{lai2025leocc}. These systems optimize throughput
or latency but do not consider protected spectrum users or interference
constraints. Our work instead treats radio astronomy protection as a
first-class constraint in satellite scheduling.

\vspace{-0.5em}
\section{Conclusion}
\label{sec:conclusion}

We present \sysabb, a constellation-wide sky sharing system that enables predictive, interference-aware scheduling of LEO satellite spot beams to protect radio astronomy, while preserving LEO network coverage. By integrating real time ODS data, high-fidelity orbital prediction, and ITU-compliant EPFD modeling, \sysabb transforms RAS protection requirements into network control decisions. Using EPFD-budgeted Region of Interest (RoI) abstraction, we formulate LEO--RAS coexistence as a MILP scheduling problem and propose a flow-based algorithm that achieves optimal solutions in special cases and scalable near-optimal solutions in the general setting. Extensive evaluation across \NSites radio astronomy sites shows that \sysabb substantially reduces service disruption compared to beam avoidance while maintaining RAS protection. These results demonstrate that constellation-wide coordination is both practical and necessary for sustainable spectrum sharing between satellite networks and radio astronomy observations.

\vspace{-0.5em}
\begin{acks}
This work is supported by NSF Awards 2235140 and 2443035.
\end{acks}

\clearpage
\newpage
\bibliographystyle{ACM-Reference-Format}
\bibliography{references}

\clearpage
\appendix

\section{Appendices}

\subsection{Starlink Downlink Antenna Model}
\label{app:star-antenna}
For Starlink, the spot beams operate over an aggregate frequency spectrum of 2000 MHz, which is divided into 8 channels of 250 MHz each. Each gen2-Mini satellite is equipped with 5 downlink phased-array antennas, with each antenna capable of illuminating 8 spot beams with 2 polarizations. Therefore, the total number of downlink beams for each satellite is 80. 
We model each Starlink downlink antenna as a $25\times40$ half-wavelength uniform rectangular array, chosen to approximate a $34$~dB boresight gain.

The gain of each spot beam resulting from a subarray of the antenna can be specified as the product of the element gain and array factor, assuming a 2D rectangular array, using the following equation:
\begin{equation} 
G_{T\text{x},s,c} [dB] = G_{Elem} [dB] + 10\log_{10}(N_x \times N_y),
\label{eq:phased-array-gain}
\end{equation}
where $G_{Elem}$ is the gain of a single element in the phased array.

\subsection{Satellite beam Capacity Model}
\label{app:spot beam-cap}

ITU regulations restrict maximum ground-level PFD in Ku-band to protect terrestrial services~\cite{ITU-R_SF.1482,ITUMaxPower}. Satellites dynamically control EIRP so that ground PFD remains near this limit regardless of elevation angle.

Let $\wedge_{s,c}$ denote the capacity of the spot beam from satellite $s$ serving cell $c$. The received power at the user terminal is
\begin{equation}
P_{R\text{x},s,c} = 10^{\text{PFD}_{max}/10} \cdot A_{eff} \cdot B,
\end{equation}
where $A_{eff}$ is the effective aperture area of the receiver antenna:
\begin{equation}
A_{eff} = \frac{G_{R\text{x}}\lambda^2}{4\pi}.
\end{equation}
The achievable downlink capacity is then
\begin{equation}
\wedge_{s,c} =
\mathbf{1}\{\text{el}_{s,c} \ge 25^\circ\}
B \log_2 \left(1 + \frac{P_{R\text{x},s,c}}{k T_{sys} (10^6 B)} \right),
\label{eq:cap_nominal}
\end{equation}
where $k$ is Boltzmann’s constant, $T_{sys}$ is receiver noise temperature, and $\text{el}_{s,c}$ is the elevation of satellite $s$ with respect to cell $c$.

This capacity model allows \sysabb to determine if a satellite can satisfy cell demand while respecting both interference and hardware constraints.

\subsection{\sysabb Practicality} \label{sec:practicality}
We discuss how \sysabb can be deployed within the operational constraints of current LEO satellite networks.

\noindent \textbf{Control-plane operation.}
\sysabb operates entirely at the satellite network control plane and requires no changes to satellite hardware or onboard beamforming logic. Scheduling is performed using a rolling-horizon approach: beam--cell assignments are periodically recomputed using predicted satellite states and ODS metadata.

Control-plane timescales are governed by beam reconfiguration latency, handover stability, and management overhead. In practice, scheduling windows on the order of several seconds to tens of seconds are sufficient to track constellation dynamics while avoiding excessive configuration churn. Because satellite motion is highly predictable, schedules can be computed and deployed proactively, before satellites enter sensitive interference regions.

\noindent \textbf{Scalability and integration.}
\sysabb scales with constellation size by restricting optimization to finite regions of interest (RoIs) surrounding observatories; cells outside protected regions continue to operate under standard network scheduling policies. The modular ODS interface allows new observatories and updated protection requirements to be incorporated without changes to satellite hardware or onboard software.
Although we present \sysabb as a centralized scheduler, the same design can be extended to hierarchical or distributed control (e.g., using inter-satellite links) without modifying the core scheduling logic.

Overall, \sysabb demonstrates that predictive, constellation wide control-plane scheduling can be integrated into existing LEO network architectures to enable practical spectrum coexistence while preserving service coverage.

\subsection{Notation Table} \label{app:notation_table}

We summarize the notation used in this study in the order they appear in the paper in Table \ref{tab:notation}--\ref{tab:notation2}. Table~\ref{tab:notation3} summarizes notation used only in the appendix (e.g., in derivations, proofs, and auxiliary algorithms).

\begin{table}[!t]
\caption{Table of Notations} \label{tab:notation}
\vspace{-0.1in}
\centering
\vspace{-0.1in}
\begin{tabular}{p{2cm}p{5.7cm}}
\hline \hline
    Symbol & Description \\
    \hline
    $S$ & Set of all satellites in the constellation\\
    $N^{sb}$ & Number of simultaneous spot beams on each satellite\\
    $C^{all}$ & Set of all cells on the ground\\
    $N_w$ & Duration which satellite spot beam--cell assignment remains fixed\\
    $L_R$ & RAS telescope coordinates\\
    $E_R$ & RAS telescope elevation\\
    $\Omega_O$ & RAS boresight right ascension\\
    $D_O$ & RAS boresight declination\\
    $F_R$ & RAS operating frequency\\
    $\Delta F$ & RAS observation Bandwidth\\
    $T_{Start}$ & RAS observation start time\\
    $T_{End}$ & RAS observation end time\\
    $\tau$ & RAS observation duration\\
    $RO$ & Radio observation data tuple\\
    $\Delta t$ & Time discretization interval length\\
    $\T$ & Set of all discretized time intervals\\
    $W$ & Set of all scheduling windows\\
    $\theta_{c,t,s}$ & Angle between cell $c$, satellite $s$, and RAS at time $t$\\
    $G_{T\text{x}}(\theta_{c,t,s})$ & Satellite downlink transmit gain at off-axis angle $\theta_{c,t,s}$\\
    $\phi_{t,s}$ & Angle between RAS boresight direction, RAS, and satellite s at time t\\
    $P_{T\text{x},c,t,s}$ & Transmit power allocated to the spot beam of satellite $s$ serving cell $c$ at time interval $t$ \\
    
    $G_{R\text{x}}(\phi_{t,s})$ & RAS receiver gain at off-axis angle $\phi_{t,s}$\\
    $G_{R\text{x},max}$ & Peak gain of the RAS receive antenna \\
    $D_{t,s}$ & Distance between satellite $s$ at time $t$ and RAS\\
    \text{PFD}($\theta_{c,t,s}$) & The power flux density (PFD) at the RAS from satellite $s$ transmitting to cell $c$ at time interval $t$\\
    $\text{EIRP}(\theta_{c,t,s})$ & The effective isotropic radiated power toward the RAS at off-axis angle $\theta_{c,t,s}$\\
    $e_{c,t,s}$ & The EPFD contribution at the RAS telescope from a single spot beam of satellite $s$ serving cell $c$ at time interval $t$\\
    $x_{c,t,s}$ & The normalized capacity of a spot beam from satellite $s$ dedicated to serving cell $c$ at time interval $t$\\
    $\text{EPFD}(t)$ & EPFD at the RAS telescope at time interval $t$ \\
    $TH_{RAS}$ & Detrimental interference level\\
    $\Delta TH_{RAS}$ & Residual EPFD budget dedicated to non-RoI cells\\
    $\overline{TH}_{RAS}$ & $TH_{RAS} - \Delta TH_{RAS}$\\
\hline 
\end{tabular}
\end{table}

\begin{table}
\caption{Table of Notations (Continuation)} \label{tab:notation2}
\vspace{-0.1in}
\centering
\vspace{-0.1in}
\begin{tabular}{p{2cm}p{5.7cm}}
\hline \hline
    Symbol & Description \\
    \hline
    $\phi_{TH}$ & Angular threshold used for excluding satellites for non-RoI, operator-default assignment policy\\
    $R$ & Distance of served cell from RAS for RoI inclusion\\ $\text{EPFD}_{\text{RoI}}^{(\phi_{TH}, R)}(t)$ & EPFD contribution from RoI cells in RoI calculation procedure\\
    $\text{EPFD}_{\text{Non-RoI}}^{(\phi_{TH}, R)}(t)$ &
    EPFD contribution from non-RoI cells in RoI calculation procedure\\
    $x_d^c$ & Upper-bound on normalized cell demand\\
    $z_{c,w,s}$ & Auxiliary variable showing if cell $c$ is assigned to satellite $s$ in window $w$.\\
    $x_{c,w,s}$ & The normalized capacity of a spot beam from satellite $s$ dedicated to serving cell $c$ at window $w$ \\
    $C_s^t$ & Set of all cells that are considered for satellite $s$ at time $t$ by \sysabb\\
    $G^f$ & Flow network that encodes all feasible satellite--cell assignments over the scheduling interval\\
    $N$ & Set of all nodes in $G^f$\\
    $A$ & Set of all edges in $G^f$\\
    $v_{c,t}$ & Cell node for cell $c$ at time $t$ in $G^f$ \\
    $u_{s,t}$ & Satellite node for satellite $s$ at time $t$ in $G^f$ \\ 
    $r$ & Source node at flow network $G^f$ \\
    $d$ & Sink node at flow network $G^f$\\
    $f(r, v_{c,t})$ & Flow between source node $r$ and cell node $v_{c,t}$ \\
    $f(v_{c,t}, u_{s,t})$ & Flow between cell node $v_{c,t}$ and sat node $u_{s,t}$\\
    $f(u_{s,t}, d)$ & Flow between satellite node and sink node\\
    $\text{EPFD}^\star(q)$ & Minimum achievable EPFD for serving exactly q cell, time pairs \\
    $q^\star$ & Maximum value of q achieving $\text{EPFD}^\star(q) < \tau/\Delta t \Delta F \overline{TH}_{RAS}$ \\
    $x_d$ & Uniform demand of cells\\
    $\overline{x_d}$ & Average cell demand when cells have heterogenous demand\\
    $v_{c,w}$ & cell node at window $w$ in $G^f$\\
    $u_{s,w}$ & Satellite node at window $w$ in $G^f$\\
    $\rho_{s,w}$ & Spot-beam utilization of satellite $s$ at window $w$\\
    $|EPFD|$ & Largest $\text{EPFD}$ value in Flow graph $G^f$ \\
    $|\overline{S}|$ & Average number of satellites overhead a cell around the RAS site.\\
    $C_{\text{RoI}}$ & Set of all cells in the RoI\\

\hline 
\end{tabular}
\end{table}

\begin{table}[!t]
\caption{Table of Notations (Appendix Additions)}
\label{tab:notation3}
\centering
\begin{tabular}{p{2cm}p{5.7cm}}
\hline \hline
Symbol & Description \\
\hline
$\wedge_{s,c}$ & Capacity of the spot beam from satellite $s$ serving cell $c$ (Appendix~\ref{app:spot beam-cap}) \\
$P_{R\text{x},s,c}$ & Received power at the user terminal for the link from satellite $s$ to cell $c$ \\
$\text{PFD}_{max}$ & Maximum ground-level PFD limit used in the spot-beam capacity model \\
$A_{\text{eff}}$ & Effective aperture area of  antenna\\
$B$ & Bandwidth of satellite spot-beam\\
$G_{R\text{x}}$ & Satellite user receiver antenna gain\\
$k$ & Boltzmann's constant \\
$T_{sys}$ & Receiver system noise temperature\\
$\text{el}_{s,c}$ & Elevation angle of satellite $s$ from cell $c$ \\
$\phi_{TH}^{\max}$ & Maximum angular threshold considered in the online RoI parameter sweep \\
$R_{\max}$ & Maximum geographic radius considered in the online RoI parameter sweep \\
$G_{T\text{x},s,c}$ & Spot-beam transmit gain for satellite $s$ serving cell $c$ in the phased-array gain model \\
$G_{Elem}$ & Gain of a single antenna element in the Starlink phased-array model \\
$N_x$ & Number of phased-array elements along the $x$ dimension \\
$N_y$ & Number of phased-array elements along the $y$ dimension \\
$f(a)$ & Flow on arc $a \in A$ in the min-cost flow formulation \\
$\text{cost}(a)$ & Cost of arc $a \in A$ in the min-cost flow formulation \\
$\texttt{cap}(i,j)$ & Capacity of arc $(i,j)$ in the heuristic and prune helper \\
$I_{\max}$ & Iteration cap for the general-case \algname heuristic \\
$\mathcal{F}$ & Forbidden set of assignments/edges excluded from subsequent heuristic iterations \\
$\mathcal{A}$ & Assignment set returned by the scaled min-cost flow in the general-case heuristic \\
$\mathcal{A}^{\star}$ & Best feasible assignment set found so far in the general-case heuristic \\
$\mathcal{A}_{\text{feas}}$ & Feasible assignment set after overload pruning in the general-case heuristic \\
$\mathcal{L}$ & Temporary list of assignments for an overloaded $(t,s)$ pair in the prune helper \\
\hline
\end{tabular}
\end{table}

\subsection{Acronym Table}
\label{app:acronyms}

We summarize the acronyms and abbreviations used in this paper in Table~\ref{tab:acronyms}.

\begin{table}[t]
\caption{Table of Acronyms and Abbreviations}
\label{tab:acronyms}
\centering
\begin{tabular}{p{2cm}p{5.7cm}}
\hline \hline
Acronym & Meaning \\
\hline
API & Application Programming Interface \\
CBRS & Citizens Broadband Radio Service \\
DTBA & Dynamic Telescope Boresight Avoidance \\
EIRP & Effective Isotropic Radiated Power \\
EPFD & Equivalent Power Flux Density \\
FCC & Federal Communications Commission \\
GPS & Global Positioning System \\
ITU & International Telecommunication Union \\
ITU-R & ITU Radiocommunication Sector \\
Jy & Jansky (unit commonly used for radio astronomy sensitivity / flux density) \\
LEO & Low Earth Orbit \\
MILP & Mixed-Integer Linear Program \\
NRAO & National Radio Astronomy Observatory \\
ODS & Operational Data Sharing \\
PFD & Power Flux Density \\
RAS & Radio Astronomy Services \\
REST & Representational State Transfer \\
RFI & Radio Frequency Interference \\
RoI & Region of Interest \\
SKA & Square Kilometre Array \\
SNR & Signal-to-Noise Ratio \\
SoTA & State of the Art \\
TBA & Telescope Boresight Avoidance \\
TLE & Two-Line Element \\
UTC & Coordinated Universal Time \\
VLBA & Very Long Baseline Array \\
\hline
\end{tabular}
\end{table}

\subsection{EPFD-budgeted Region of Interest (RoI) Algorithm}
\label{app:roi}

Algorithm~\ref{alg:roi-procedure} summarizes the online procedure used to select $(\phi_{TH}, R)$, and the corresponding $\Delta TH_{RAS}$, for a given RAS site,  constellation, and observation configuration. The inputs are the candidate ranges for $\phi_{TH}$ and $R$, the set of time intervals that span $\tau$, i.e., $\T$, the EPFD limit $TH_{RAS}$, Bandwidth $\Delta F$, and observation ODS data. The output is the first feasible $(R,\phi_{TH},\Delta TH_{RAS})$ under the prescribed search order, or a failure return if no feasible configuration is found. The procedure sweeps candidate pairs in increasing order, runs \sysabb over $\tau$, computes the EPFD slack, and returns the implied $\Delta TH_{RAS}$ from the non-RoI EPFD when the slack is positive. Its complexity is linear in the number of tested $(\phi_{TH}, R)$ candidate pairs and is dominated by the repeated \algname executions over $\tau$.

\begin{algorithm}[h]
\caption{Online RoI calculation procedure for selecting $(\phi_{TH}, R, \Delta TH_{RAS})$}
\label{alg:roi-procedure}
\small
\begin{algorithmic}[1]
\REQUIRE Candidate ranges $(0,\phi_{TH}^{\max}]$ and $(0,R_{\max}]$, representative time set $\T$, EPFD limit $TH_{RAS}$, bandwidth $\Delta F$, ODS data
\ENSURE $(R,\phi_{TH},\Delta TH_{RAS})$ if feasible; otherwise “no feasible configuration found”
\FOR{$R$ in increasing order over $(0,R_{\max}]$}
    \FOR{$\phi_{TH}$ in increasing order over $(0,\phi_{TH}^{\max}]$}
        \STATE Run \sysabb over all $t \in \T$ with cells inside the RoI participating in the optimization, and when a satellite is within $\phi_{TH}$ of the RAS boresight, also include all cells visible to that satellite; record $\text{EPFD}_{\text{RoI}}^{(\phi_{TH}, R)}(t)$ and $\text{EPFD}_{\text{Non-RoI}}^{(\phi_{TH}, R)}(t)$
        \STATE Set
        \[
        \text{slack} \;=\; TH_{RAS} - \frac{1}{\Delta F}\max_{t \in \T}\Bigl(\text{EPFD}_{\text{RoI}}^{(\phi_{TH}, R)}(t) + \text{EPFD}_{\text{Non-RoI}}^{(\phi_{TH}, R)}(t)\Bigr)
        \]
        \IF{$\text{slack} > 0$}
            \STATE Set
            \[
            \Delta TH_{RAS} \;=\; \frac{1}{\Delta F}\max_{t \in \T}\text{EPFD}_{\text{Non-RoI}}^{(\phi_{TH}, R)}(t)
            \]
            \STATE \textbf{return} $(R,\phi_{TH},\Delta TH_{RAS})$
        \ENDIF
    \ENDFOR
\ENDFOR
\STATE \textbf{return} “no feasible configuration found”
\end{algorithmic}
\end{algorithm}

\subsection{NP-hardness of Problem \textbf{P0}}
\label{app:nphard}

\begin{proof}[Proof of Theorem~\ref{theorem:nphard}]
    We prove NP-hardness by a reduction from the decision version of Bin Packing. Consider an instance with item sizes $\{a_c\}_{c\in C_s^t}$, bin capacity $\kappa$, and $|S|$ bins. Construct an instance of \textbf{P0} with a single time slot $|\T|=1$ (hence a single window), set $x_d^c=a_c$ for all $c\in C_s^t$, set $N^{sb}=\kappa$, and set $\text{e}_{c,t,s}=0$ for all $(c,t,s)$ so that the EPFD constraint is non-binding. Because $x_{c,t,s}\in\{0\}\cup x_d^c$ and $x_{c,t,s}\le x_d^c z_{c,\lfloor t/N_w \rfloor,s}$ with $z_{c,w,s}\in\{0,1\}$, any feasible assignment satisfies $z_{c,w,s}=1 \Rightarrow x_{c,t,s}\in\{0,x_d^c\}$ and $z_{c,w,s}=0 \Rightarrow x_{c,t,s}=0$; together with $\sum_{s\in S}z_{c,w,s}\le 1$ and $\sum_{s\in S}x_{c,t,s}\le x_d^c$, each $c$ can be assigned to at most one $s$ and, if assigned, must have $x_{c,t,s}=x_d^c$. The satellite-capacity constraint $\sum_c x_{c,t,s}\le N^{sb}$ then enforces $\sum_{c\ \text{assigned to}\ s} a_c \le \kappa$ for every $s\in S$. Finally, the objective $\sum_{c,t,s} x_{c,t,s}/x_d^c$ counts the number of served cells (each term is $0$ or $1$), so achieving value $|C_s^t|$ is possible if and only if all items can be packed into the $|S|$ bins of capacity $\kappa$. Hence Bin Packing reduces to deciding whether \textbf{P0} attains objective at least $|C_s^t|$, implying \textbf{P0} is NP-hard.
\end{proof}

\subsection{Formal definition of \(\textbf{P}_1(q)\)}
\label{app:p1q}

For the special case \(N_w=1\) and \(x_d^c=1\), \(\textbf{P}_1(q)\) is defined as a minimum-interference assignment problem: 

\begin{align}
    \textbf{P1(q): } \min_{x} \quad & \sum_{t \in \T} \sum_{s \in S} \sum_{c \in C_s^t} \text{e}_{c,t,s} x_{c,t,s} \label{eq:mincost-obj}\\
    \text{s.t.} \quad
    & \sum_{s \in S} x_{c,t,s} \leq 1 \quad && \forall\, c\in C^{all},\; t \in \T \label{eq:mincost-c1}\\
    & \sum_{c \in C_s^t} x_{c,t,s} \leq N^{sb} \quad && \forall\, t \in \T,\; s \in S \label{eq:mincost-c2}\\
    & \sum_{t \in \T} \sum_{s \in S} \sum_{c \in C_s^t}  x_{c,t,s} = q   \label{eq:mincost-c3}\\
    & x_{c,t,s} \in \{0, 1\} \label{eq:mincost-c4}
 \end{align}

\subsection{Proofs}
\label{app:proofs}

\subsubsection{Proof of Lemma~\ref{lemma:monotone}, Monotonicity of $\text{EPFD}^\star(q)$)}
\label{app:non-decreasing}
\begin{proof}
Assume by way of contradiction there exists $q$, $q'$, with $q < q'$ such that $\text{EPFD}^\star(q') < \text{EPFD}^\star(q)$. Now select any $(q'-q)$ active elements from the solution for $q'$, and remove them; the result will be a feasible assignment for q; call the remaining interference $\hat{\text{EPFD}}(q)$. Since interference values are non-negative, we get: $\hat{\text{EPFD}}(q) \leq \text{EPFD}^\star(q') < \text{EPFD}^\star(q)$, which results in $\hat{\text{EPFD}}(q) < \text{EPFD}^\star(q)$, which is impossible since $\text{EPFD}^\star(q)$ is the minimum achievable $\text{EPFD}$ for $q$. Hence, the contradiction assumption is wrong, and $\text{EPFD}^\star(q)$ is non-decreasing in $q$.    
\end{proof}

\subsubsection{Proof of Theorem \ref{theorem:q-star-optimal}, Optimality of $q^\star$}
\label{app:theorem6-2}
\begin{proof}
Feasibility: for every $q \leq q^\star$, a solution to problem $\textbf{P}_1(q)$ such that
$\text{EPFD}^\star(q) <\Delta F\frac{\tau}{\Delta t}\,\overline{TH}_{RAS}$
satisfies every constraint in the original problem \textbf{P0}, hence it is also a feasible solution to our original problem \textbf{P0}.

Optimality: For any $q > q^\star$, from \text{\large L}\text{\small EMMA}~\ref{lemma:monotone}, we have $\text{EPFD}^\star(q) \geq \text{EPFD}^\star(q^\star)$. Also, from the definition of $q^\star$ we can see that $\text{EPFD}^\star(q) > \Delta F\frac{\tau}{\Delta t}\,\overline{TH}_{RAS}$. Hence, no $q > q^\star$ results in a feasible solution.
As a result, the problem \textbf{P0} cannot achieve more than $q^\star$, and it achieves $q^\star$, completing the proof.
\end{proof}

\subsubsection{Proof of Lemma~\ref{lemma:one-to-one}}
\label{app:one-to-one}

\begin{proof}
    Given (x), create a unit of flow
    along the path (r to $v_{c,t}$ to $u_{s,t}$ to $d$)
    for every triple with $x_{c,t,s}=1$.
    Because Eq.~\ref{eq:mincost-c3} forces exactly
    (q) Such triples, supply/demand is respected.
    Constraint (\ref{eq:mincost-c1}) implies $\text{cap}(r\to v_{c,t})$ is never exceeded,
    and constraint (\ref{eq:mincost-c2}) results in
    \(
    \sum_{c} x_{c,t,s}\le N^{sb}
    \Rightarrow
    \text{cap}(u_{s,t}\to q)\) never being exceeded.
    Hence, the resulting (f) is a feasible flow. Conversely, an integral (q)-unit flow must send exactly
    one unit through each arc (r to $v_{c,t}$) it uses
    (because the capacity is 1) and can only exit
    ($v_{c,t}$) via a single arc to some ($u_{s,t}$).
    Define ($x_{c,t,s}=1$) if that arc carries flow,
    and (0) otherwise.
    Arc capacities guarantee constraints (\ref{eq:mincost-c1}) and (\ref{eq:mincost-c2}) of Problem $\textbf{P}_1(q)$. Finally, the integrality of the flow implies
    $x_{c,t,s} \in \{0,1\}$.
\end{proof}

\subsubsection{Proof of Theorem~\ref{theorem:equiv}}
\label{app:flow_equiv}
\begin{proof}
    By \text{\large L}\text{\small EMMA}~\ref{lemma:one-to-one}, every feasible assignment $x$ corresponds to a unique integral flow $f$, and vice versa.
    Arc $v_{c,t}\to u_{s,t}$ bears cost
    $\text{e}_{c,t,s}$; therefore    
    \[
    \text{cost}(f)=
    \sum_{c,t,s}\text{e}_{c,t,s}\,x_{c,t,s}.
    \]    
    The objective functions coincide; hence, it suffices to show that the optimal solution to the Min-Cost Flow problem on $G^f$ is integral. From \cite{AhujaMagnantiOrlin1993-Theorem9.10}, the optimal solution to the Min-Cost Flow problem for integer demand and capacity is integral; since all capacities and demands are integral in Flow Network $G^f$, we get $\min_f\text{cost}(f)=\min_X\text{EPFD}(X)=\text{EPFD}^\star(q)$.
\end{proof}

\subsection{Algorithms}
\label{app:alg}

\subsubsection{\textbf{\algname for the special case.}}
Algorithm~\ref{alg:quietsat} gives the full \algname procedure for the special case $N_w=1$ and $x_d^c=1$. \textit{Inputs} are the constellation TLEs, the candidate cell sets $C_s^t$, the time set $\T$, the satellite set $S$, the observatory tuple $\text{RO}$, and the EPFD budget $\Delta F\frac{\tau}{\Delta t}\overline{TH}_{RAS}$. \textit{Outputs} are the assignment $x^\star$ and the maximum feasible served count $q^\star$. \textit{Key steps} are: (i) pre-computing the EPFD table $\text{e}_{c,t,s}$ for all visible satellite--cell pairs, (ii) binary search over $q$, and (iii) invoking Algorithm~\ref{alg:mcf} to test feasibility under the EPFD budget and update $(x^\star,q^\star)$.

\begin{algorithm}[h]
\caption{\algname Algorithm for $N_w=1, x_d^c=1$}
\label{alg:quietsat}
\small
\begin{algorithmic}[1]
\REQUIRE TLE files for constellation; cell set $C_s^t$; Time intervals set $\T$;
         satellite set $S$; window length $N_w$;\\
         Radio Observatory tuple $\text{RO}$;
         \\Interference limit $\Delta F\frac{\tau}{\Delta t} \overline{TH}_{RAS}$
\STATE \textbf{Pre‑processing: generate EPFD table}\label{line:preproc}
\FORALL{$t\in \T$}
   \FORALL{$s\in S$}
      \STATE Propagate TLE of $s$ to time interval $t$
   \ENDFOR
   \FORALL{$s \in S$}
       \FORALL{$c\in C_s^t$}  
         \STATE $\phi_{t,s}\gets$ off‑axis angle
                $(s,\text{RO.\,dir})$
         \STATE $\theta_{c,t,s}\gets$ off‑axis angle
                $(s,c,\text{RO})$
         \STATE $G_{\text{R\text{x}}}\!\leftarrow
                G_{R\text{x}}(\phi_{t,s})$  \COMMENT{RA antenna}
         \STATE $G_{\text{T\text{x}}}\!\leftarrow
                G_{\text{T\text{x}}}(\theta_{c,t,s})$
                \COMMENT{Satellite Antenna}
         \STATE $\text{e}_{c,t,s} \gets
                \text{PFD}(\theta_{c,t,s})\,
                G_{R\text{x}}(\phi_{t,s})/G_{R\text{x},max}$
                \hfill
      \ENDFOR
   \ENDFOR
\ENDFOR

\STATE $L\gets0$;\; $U\gets \max_{s,t}(|C_s^t|\,|\mathcal{T}|$);\;
       $q^\star\gets0$
\WHILE{$L\le U$}                                
   \STATE $q\gets\lfloor(L+U)/2\rfloor$
   \STATE $(\text{cost},x)\gets\textsc{MinCostFlow}(q)$
   \IF{$\text{cost}\le \Delta F\frac{\tau}{\Delta t} . \overline{TH}_{RAS}$}
       \STATE $q^\star\gets q$;\; $x^\star\gets x$;\; $L\gets q+1$
   \ELSE
       \STATE $U\gets q-1$
   \ENDIF
\ENDWHILE
\RETURN $(x^\star,q^\star)$
\end{algorithmic}
\end{algorithm}


\subsubsection{\textbf{Min-cost flow subroutine for fixed $q$.}}
Algorithm~\ref{alg:mcf} solves the fixed-demand subproblem used by Algorithm~\ref{alg:quietsat}. \textit{Inputs} are the flow demand $q$ and the pre-computed data (including $\text{e}_{c,t,s}$ and the visible satellite--cell pairs). \textit{Outputs} are $\text{EPFD}^\star(q)$ and the corresponding assignment $x$. \textit{Key steps} are: constructing the flow graph $G^f$, setting source/sink supplies to $q$, solving the min-cost flow instance, and extracting $x_{c,t,s}$ from the resulting flow paths. \textit{Complexity} consists of graph construction and extraction (linear in the graph size) and one min-cost-flow solve, which is the dominant step; the overall solver-dependent cost is the same min-cost-flow complexity used in the time-complexity discussion in \S\ref{sub:algoverview}.

\begin{algorithm}[h]
\caption{\textsc{MinCostFlow}$(q)$}
\label{alg:mcf}
\small
\begin{algorithmic}[1]
\REQUIRE Flow demand $q$ and the data used by Algorithm~\ref{alg:quietsat}
\ENSURE Minimum cost $\text{EPFD}^\star(q)$ and assignment $x$
\STATE Build the network flow graph $G^f = (N, A)$ 
\STATE Set supplies: $\text{supply}(r)\gets +q$, $\text{supply}(d)\gets -q$
\STATE $f \gets \text{MCFSolver}(G^f)$
\STATE $\text{EPFD}^\star(q) \gets \sum_{a\in A} cost(a)\,f(a)$
\STATE Extract $x_{c,t,s}$ from every path  
       $r \rightarrow v_{c,t} \rightarrow u_{s,t} \rightarrow d$  
       that carries one unit of flow in $f$
\RETURN $(\text{EPFD}^\star(q)$,$x$)
\end{algorithmic}
\end{algorithm}

\subsubsection{\textbf{\algname heuristic for arbitrary demand.}}
Algorithm~\ref{alg:quietsat_xdc_heuristic} describes the \algname heuristic for the general case with arbitrary $x_d^c$. \textit{Inputs} are candidate edges $(c,t,s)$, demands $\{x_d^c\}$, the spot-beam limit $N^{sb}$, EPFD costs $\{\text{e}_{c,t,s}\}$, the total EPFD budget in \textbf{P0}, and the iteration cap $I_{\max}$. \textit{Outputs} are a feasible assignment $(x,z)$ for \textbf{P0}. \textit{Key steps} are: constructing a scaled flow graph, solving a scaled min-cost flow, calling Algorithm~\ref{alg:quietsat_xdc_prune} to repair satellite overloads, updating satellite-side capacities based on observed utilization, and repeating until convergence or the iteration cap is reached. \textit{Complexity} is dominated by at most $I_{\max}$ scaled min-cost-flow solves, with additional per-iteration repair and capacity-updates.

\label{app:heuristic}
\begin{algorithm}[h]
\caption{\algname heuristic with arbitrary $x_d^c,\;N_w=1$}
\label{alg:quietsat_xdc_heuristic}
\small
\begin{algorithmic}[1]
\REQUIRE Candidate edges $(c,t,s)$, demands $\{x_d^c\}$, spot beam limit $N^{sb}$, EPFD costs $\{\text{e}_{c,t,s}\}$, total EPFD budget in \textbf{P0}, iteration cap $I_{\max}$.
\ENSURE Feasible assignment $(x,z)$ for \textbf{P0}.

\STATE Construct scaled network; set $\texttt{cap}(r,v_{c,t})\leftarrow 1$, $\forall(c,t)$.
\STATE Set $cost(v_{c,t},u_{s,t}) \leftarrow x_d^c\,\text{e}_{c,t,s}$, $\forall(c,t,s)$.
\STATE Compute $\overline{x_d}$ and set $\texttt{cap}(u_{s,t},d)\leftarrow \left\lfloor \dfrac{N^{sb}}{\overline{x_d}} \right\rfloor$, $\forall(t,s)$.
\STATE Initialize $\mathcal{F}\leftarrow \emptyset$ and $\mathcal{A}^\star \leftarrow \emptyset$.

\FOR{$i=1$ \TO $I_{\max}$} 
    \STATE Solve the scaled min-cost flow under the EPFD budget on the graph excluding edges in $\mathcal{F}$, obtaining $\mathcal{A}=\{(c,t,s)\}$.
    \STATE $(\mathcal{A}_{\text{feas}},\mathcal{F}) \leftarrow \textsc{PruneOverload}(\mathcal{A},\mathcal{F},\{x_d^c\},N^{sb},\{\text{e}_{c,t,s}\})$.
    \IF{$|\mathcal{A}_{\text{feas}}| > |\mathcal{A}^\star|$}
        \STATE $\mathcal{A}^\star \leftarrow \mathcal{A}_{\text{feas}}$.
    \ENDIF

    \STATE Compute $\rho_{s,t}\leftarrow \sum_{(c,t,s)\in \mathcal{A}} x_d^c$, $\forall(t,s)$.
    \FORALL{$(t,s)$}
        \IF{$\rho_{s,t} > N^{sb}$}
            \STATE $\texttt{cap}(u_{s,t},d) \leftarrow \min\!\Big\{\texttt{cap}(u_{s,t},d),\; \big|\{(c,t,s)\in \mathcal{A}_{\text{feas}}\}\big|\Big\}$.
        \ELSE
            \STATE $\texttt{cap}(u_{s,t},d) \leftarrow \texttt{cap}(u_{s,t},d) + \left\lfloor \dfrac{N^{sb}-\rho_{s,t}}{\overline{x_d}} \right\rfloor$.
        \ENDIF
    \ENDFOR

    \IF{$\mathcal{A}_{\text{feas}}=\mathcal{A}$}
        \STATE \textbf{break}
    \ENDIF
\ENDFOR

\STATE Set $z_{c,t,s}=1$ iff $(c,t,s)\in \mathcal{A}^\star$, else $0$.
\STATE Set $x_{c,t,s}=x_d^c$ iff $(c,t,s)\in \mathcal{A}^\star$, else $0$.
\RETURN $(x,z)$.
\end{algorithmic}
\end{algorithm}

\subsubsection{\textbf{Feasibility-repair helper.}}
Algorithm~\ref{alg:quietsat_xdc_prune} is the repair helper used by Algorithm~\ref{alg:quietsat_xdc_heuristic} to enforce the true spot-beam limit after a scaled flow solve. \textit{Inputs} are the selected assignment set $\mathcal{A}$, the forbidden set $\mathcal{F}$, demands $\{x_d^c\}$, the spot-beam limit $N^{sb}$, and EPFD costs $\{\text{e}_{c,t,s}\}$. \textit{Outputs} are the repaired feasible set $\mathcal{A}_{\text{feas}}$ and the updated forbidden set $\mathcal{F}$. \textit{Key steps} are: identifying overloaded $(t,s)$ pairs, sorting the selected assignments for each overloaded pair by $\text{e}_{c,t,s}$, and removing assignments until the total demand fits within $N^{sb}$. \textit{Complexity} is a scan over selected assignments plus sorting within overloaded $(t,s)$ groups.

\begin{algorithm}[!t]
\caption{\algname helper: prune overloaded $(t,s)$ pairs}
\label{alg:quietsat_xdc_prune}
\small
\begin{algorithmic}[1]
\REQUIRE $\mathcal{A}$, forbidden set $\mathcal{F}$, demands $\{x_d^c\}$, limit $N^{sb}$, costs $\{\text{e}_{c,t,s}\}$.
\ENSURE $(\mathcal{A}_{\text{feas}},\mathcal{F})$.

\STATE $\mathcal{A}_{\text{feas}} \leftarrow \mathcal{A}$.
\FORALL{$(t,s)$ with $\sum_{(c,t,s)\in \mathcal{A}_{\text{feas}}} x_d^c > N^{sb}$}
    \STATE $\mathcal{L}\leftarrow \{(c,t,s)\in \mathcal{A}_{\text{feas}}\}$.
    \STATE Order $\mathcal{L}$ by increasing $\text{e}_{c,t,s}$ (lowest first).
    \WHILE{$\sum_{(c,t,s)\in \mathcal{L}} x_d^c > N^{sb}$}
        \STATE Remove the next $(c,t,s)$ from $\mathcal{A}_{\text{feas}}$ and add it to $\mathcal{F}$.
        \STATE $\mathcal{L}\leftarrow \mathcal{L}\setminus \{(c,t,s)\}$.
    \ENDWHILE
\ENDFOR
\RETURN $(\mathcal{A}_{\text{feas}},\mathcal{F})$.
\end{algorithmic}
\end{algorithm}

\subsection{Background on CelesTrak Element Sets}
\label{app:leo_supgp_background}

Supplemental General Perturbations (SupGP) are provided by CelesTrak with element sets (GPEs) derived directly from owner- and operator-supplied orbital data. All GP data are generated from radar and optical observations collected by the U.S. Space Surveillance Network (SSN) \cite{snn_spaceforce}; however, the limited geographic coverage of optical sensors and weather-induced observation gaps can lead to delayed catalog updates and degraded orbital accuracy. To mitigate these limitations, CelesTrak integrates operator-supplied ephemeris data into its standard update pipeline by fitting them to GPEs using the Simplified General Perturbations 4 (SGP4) orbital model implemented in Satellite Tool Kit (STK) \cite{vallado2006revisiting}. CelesTrak conducted real testbed experiments to validate the satellite tracking accuracy and
demonstrated \textit{centimeter-level accuracy} \cite{Celestrak}.

\subsection{Network Parameters of LEO Satellite Simulation} \label{app:table}

We summarize the network parameters used in this study in Table \ref{Tab2}.

\begin{table}[!b]
\caption{Network Parameters of LEO Satellites} \label{Tab2}
\centering
\begin{tabular}{p{4.2cm}p{3.5cm}}
\hline \hline
    Parameter Settings & Values \\
    \hline
    Satellite Constellation & Starlink Gen2-mini\\
    Constellation size  &  $9521$ satellites\\
    Gen2 satellite capacity & $96$ Gbps \\
    Channel division & 8 channels \\
    Bandwidth & $250$ MHz per channel\\
    Frequency reuse factor & $2$ polarizations\\
    Max EIRP & $44.8$ dBw \\
    Main lobe antenna gain & $34$ dBi\\
    Spot beams of each satellite & $80$\\
    Starlink DL antenna array size & $25 \times 40$\\
    Cell area & $252.9 \text{ km}^2$\\
    Discretized interval $\Delta t$ & $1$ s \\
    Default RAS latitude & $40^\circ$ \\
    Default RAS diameter & $30$ m \\
\hline 
\end{tabular}
\end{table}

\end{document}